\documentclass[12pt]{elsarticle}
\usepackage[margin=2.5cm]{geometry}
\usepackage{lipsum}

\makeatletter
\def\ps@pprintTitle{%
   \let\@oddhead\@empty
   \let\@evenhead\@empty
   \def\@oddfoot{\reset@font\hfil\thepage\hfil}
   \let\@evenfoot\@oddfoot
}
\makeatother

\usepackage{amsfonts,color,morefloats,pslatex}
\usepackage{amssymb,amsthm, amsmath,latexsym}
\usepackage{array}

\newtheorem{remark}{Remark}
\newtheorem{theorem}{Theorem}
\newtheorem{lemma}[theorem]{Lemma}

\newtheorem{definition}[theorem]{Definition}

\newtheorem{conj}[theorem]{Conjecture}

\newcommand{\rank}{{\mathrm{rank}}}

\newcommand{\tr}{{\mathrm{Tr}}}

\newcommand{\gf}{{\mathbb{F}}}

\newcommand{\support}{{\mathrm{suppt}}}

\newcommand{\wt}{{\mathtt{wt}}}

\newcommand{\C}{{\mathcal{C}}}

\newcommand{\cC}{{\mathcal{C}}}

\newcommand{\bc}{{\mathbf{c}}}
\newcommand{\bg}{{\mathbf{g}}}

\usepackage{lineno}
\usepackage{blindtext}
\begin{document}

\begin{frontmatter}



\title{New infinite families of near MDS codes holding $t$-designs}

\tnotetext[fn1]{
This research was supported in part by the National Natural
Science Foundation of China under Grant 12271059, in part by
the Young Talent Fund of University Association for Science and
Technology in Shaanxi, China, under Grant 20200505, and in part
by the Fundamental Research Funds for the Central Universities,
CHD, under Grant 300102122202.
}

\author{Ziling Heng}
\ead{zilingheng@chd.edu.cn}
\author{Xinran Wang}
\ead{wangxr203@163.com}

\cortext[cor]{Corresponding author}
\address{School of Science, Chang'an University, Xi'an 710064, China}





\begin{abstract}
In ``Infinite families of near MDS codes holding $t$-designs, IEEE Trans. Inform.
Theory, 2020, 66(9), pp. 5419-5428'', Ding and Tang made a breakthrough in constructing the first two infinite families of NMDS codes holding
$2$-designs or $3$-designs. Up to now, there are only a few known infinite families of NMDS codes holding $t$-designs
in the literature. The objective of this paper is to construct new infinite families of NMDS codes holding $t$-designs.
We determine the weight enumerators of the NMDS codes and prove that the NMDS codes hold $2$-designs or $3$-designs.
Compared with known $t$-designs from NMDS codes, ours have different parameters.
Besides, several infinite families of optimal locally recoverable codes are also derived via the NMDS codes.
\end{abstract}


\end{frontmatter}
\section{Introduction}
Let $q$ be a power of a prime $p$. Denote by $\gf_q$ the finite field with $q$ elements and $\mathbb{F}_{q}^*=\mathbb{F}_{q}\setminus \{0\}$.
\subsection{Linear codes}
For a positive integer $n$, a non-empty subset $\cC$ of $\gf_q^n$ is called an $[n,\kappa,d]$ \emph{linear code} over $\gf_q$ provided that it is a  $\kappa$-dimensional linear subspace of $\gf_q^n$, where $d$ is its minimum distance. Define the \emph{dual} of an $[n,\kappa]$ linear code $\cC$ over $\gf_q$ by
$$
\cC^{\perp}=\left\{ \textbf{u} \in \mathbb{F}_{q}^{n}: \langle  \textbf{u}, \mathbf{c} \rangle=0\mbox{ for all }\mathbf{c} \in \cC \right\},
$$
where $\langle \cdot,\cdot \rangle$ denotes the standard inner product. By definition, $\cC^\perp$ is an $[n,n-\kappa]$ linear code over $\gf_q$.
For an $[n,\kappa]$ linear code $\cC$ over $\gf_q$, let $A_{i}$ represent the number of codewords with weight $i$ in $\cC$, where $0\leq i \leq n$.
Define the \emph{weight enumerator} of $\cC$ as the following polynomial:
$$A(z)=1+A_{1}z+A_{2}z^2+ \cdots +A_{n}z^n.$$
The\emph{ weight distribution} of $\cC$ is defined by the sequence $(A_0,A_1,\ldots,A_n)$.
In recent years, the weight distribution of linear codes has been widely investigated in the literature \cite{T1, T2, T3, T4, T5, DWTW, DWTW2, W2, HLDC, W6, NMDS1, NMDS2, W3,W4}.
The weight distribution of a linear code contains crucial information including the error detection and correction capabilities of the code
and allows to compute the error probability of its error detection and correction \cite{K}.

In coding theory, modifying existing codes may yield interesting new codes.
A longer code can be constructed by adding a coordinate.
Let $\cC$ be an  $[n,\kappa,d]$ linear code over $\gf_q$.
 The \emph{extended code} $\overline{\mathcal{C}}$ of $\mathcal{C}$  is defined by
$$\overline{\cC}=\left\{(c_1,c_2,\ldots,c_{n},c_{n+1})\in \gf_q^{n+1}:(c_1,c_2,\ldots,c_{n})\in \cC\mbox{ with }\sum_{i=1}^{n+1}c_i=0\right\}.$$
This construction is said to be adding  an overall parity check \cite{HP}. Note that $\overline{\cC}$ has only even-like vectors.
Then $\overline{\cC}$ is also a linear code over $\gf_q$ with parameters $[n+1,\kappa,\overline{d}]$, where $\overline{d}=d$ or $d+1$.
For instance, the extended code of the binary $[7,4,3]$ Hamming code has parameters $[8,4,4]$.
The extension technique was used in \cite{D, WQY} to obtain desirable codes.

An $[n,\kappa,d]$ linear code is said to be good if it has both large rate $\kappa/n$ and large minimum distance $d$.
However, there is a tradeoff among the parameters $n,\kappa$ and $d$.
If an $[n,\kappa,d]$ linear code over $\gf_q$ exists, then the following Singleton bound holds:
$$d\leq n-\kappa+1.$$
Linear codes achieving the Singleton bound with parameters $[n,\kappa,n-\kappa+1]$
are called maximum distance separable (MDS for short) codes.
Linear codes nearly achieving the Singleton bound are also interesting and
have attracted the attention of many researchers.
An $[n,\kappa,n-\kappa]$ linear code is said to be almost maximum distance separable (AMDS for short).
It is known that the dual of an AMDS code may not be AMDS.
AMDS codes whose duals are also AMDS are said to be near maximum distance separable (NMDS for short).
NMDS codes are of interest because they have many nice applications in
finite geometry, combinatorial designs, locally recoverable codes and many other fields
\cite{T2, NMDS1, NMDS2, TANP, T7}.
In general, constructing infinite families of NMDS codes with desirable weight distribution
is challenging.
In recent years, a few families of NMDS codes
were constructed in \cite{T2, NMDS1, NMDS2, T7, WQY, 32D} and their weight distributions were determined.

\subsection{Combinatorial designs from linear codes}

Let $k,t,n$ be positive integers such that $1 \leq t \leq k \leq n$.
 Let $\mathcal{P}$ be a set with $|\mathcal{P}|=n\geq 1$.
 Denote by $\mathcal{B}$ a collection of $k$-subsets of $\mathcal{P}$.
 For each $t$-subset of $\mathcal{P}$, if there exist exactly $\lambda$ elements of $\mathcal{B}$ such that they contain this $t$-subset,
 then the pair $\mathbb{D}:=(\mathcal{P},\mathcal{B})$ is referred to as a $t$-$(n,k,\lambda)$ \emph{design} ($t$-\emph{design} for short).
 The elements in $\mathcal{P}$ and $\mathcal{B}$ are called points and blocks, respectively.
 A $t$-design is said to be \emph{simple} if it contains no repeated blocks.
 A $t$-design with $k=t$ or $k=n$ is said to be trivial.
 In this paper, we are interested in only simple and nontrivial  $t$-designs with $n>k>t$.  A $t$-$(n,k,\lambda)$ design satisfying $t\geq2$ and $\lambda=1$
 is called a \emph{Steiner system} denoted by $S(t,k,n)$.  For a $t$-$(n, k, \lambda)$ design, the following equality holds \cite{HP}:
$$\binom{n}{t}\lambda=\binom{k}{t}b,$$
where $b$ is the number of blocks in $\mathcal{B}$.
Let $\mathcal{B}^c$ represent the set of the complements of the blocks in $\mathcal{B}$. Then  $(\mathcal{P},\mathcal{B}^c)$ is a $t$-$(n,n-k,\lambda_{0,t})$ design if $(\mathcal{P},\mathcal{B})$ is a $t$-$(n,k,\lambda)$ design, where
\begin{eqnarray}\label{eqn-complementarydesign}
\lambda_{0,t}=\lambda\frac{\binom{n-t}{k}}{\binom{n-t}{k-t}}.
\end{eqnarray}
The pair $(\mathcal{P},\mathcal{B}^c)$ is referred to as the complementary design of $(\mathcal{P}, \mathcal{B})$.

 In past decades, the interplay between linear codes and $t$-designs has been a very interesting research topic.
 For one thing, the incidence matrix of a $t$-design yields a linear code. See \cite{Dingbook15} for progress in this direction.
 For another thing, linear codes may hold $t$-designs.
 The  well-known coding-theoretic construction of $t$-designs is described as follows.
 Let $\cC$ be an $[n,\kappa]$ linear code over $\gf_q$ and the  coordinates of a codeword in it be indexed by $(1,2,\ldots,n)$. Denote by $\mathcal{P}(\cC)=\{ 1,2,\ldots,n \}$.
 For a codeword $\textbf{c}=\{ c_1,c_2,\ldots,c_n \}\in \cC$, define its \emph{support}  by
 $$\support(\textbf{c})=\{1 \leq i \leq n : c_i \neq 0 \}.$$
 Denote by
 $$\mathcal{B}_{w}(\cC)=\frac{1}{q-1}\{\{\support(\bc):\wt(\bc)=w\mbox{ and }\bc\in \cC\}\},$$
where $\{\{\}\}$ denotes the multiset notation, $\wt(\bc)$ is the Hamming weight of $\bc$ and $\frac{1}{q-1}S$ denotes the multiset obtained by dividing the multiplicity of each element in the multiset $S$ by $q-1$ \cite{TDX}. If the pair $(\mathcal{P}(\mathcal{C}),\mathcal{B}_{w}(\cC))$ is a $t$-$(n,w,\lambda)$ design with $b$ blocks for $0\leq w \leq n$, we say that the code $\cC$ supports $t$-designs, where
\begin{eqnarray}\label{eqn-t}
b=\frac{1}{q-1}A_w,\ \lambda=\frac{\binom{w}{t}}{(q-1)\binom{n}{t}}A_w.
\end{eqnarray}

The following Assmus-Mattson Theorem provides a sufficient condition for a linear code to hold $t$-designs.

\begin{theorem}\label{the-AM}\cite[Assmus-Mattson Theorem]{D}
Let $\cC$ be an $[n,\kappa,d]$ linear code over $\gf_q$ whose  weight distribution is denoted by $(1,A_1,A_2,\ldots,A_n)$. Let $d^\bot$ be the minimum weight of $\cC^{\bot}$ whose weight distribution is denoted by $(1,A_1^\bot,A_2^\bot,\ldots,A_n^\bot)$. Let $t$ be an integer satisfying $1\leq t <\min\{d,d^\bot\}$. Assume that there are at most $d^\bot-t$ nonzero weights of $\cC$ in the range $\{1,2,\ldots,n-t\}$. Then the followings hold:
\begin{enumerate}
\item $(\mathcal{P}(\mathcal{C}),\mathcal{B}_{i}(\C))$ is a simple $t$-design if $A_i\neq 0$ with $d\leq i \leq w$, where $w$ is the largest integer satisfying $w\leq n$ and
    $$w-\left\lfloor\frac{w+q-2}{q-1}\right\rfloor<d.$$
\item $(\mathcal{P}(\mathcal{C}^\bot),\mathcal{B}_{i}(\C^\bot))$ is a simple $t$-design if $A_i^\bot\neq 0$ with $d^\bot\leq i \leq w^\bot$, where $w^\bot$ is the largest integer satsifying $w^\bot\leq n$ and
    $$w^\bot-\left\lfloor\frac{w^\bot+q-2}{q-1}\right\rfloor<d^\bot.$$
\end{enumerate}
\end{theorem}

The Assmus-Mattson Theorem is a powerful tool for constructing $t$-designs from linear codes and has been widely used in \cite{D, T3, T4, T9}.
Another method to prove that a linear code holds $t$-designs is via the automorphism group of the code.
Several infinite families of $t$-designs were constructed via this method in the literature \cite{D, T9, T5, DWTW, DWTW2, TDX2, T8, T10}.
The third method is directly characterizing the supports of the codewords of fixed weight.
See \cite{T2, T7, T8} for known $t$-designs obtained by this method.
Recently, Tang, Ding and Xiong generalized the Assmus-Mattson Theorem and derived $t$-designs from codes which don't satisfy
the conditions in the Assmus-Mattson Theorem and don't admit $t$-homogeneous group as a subgroup of their automorphisms \cite{TDX}.

\subsection{Motivations and objectives of this work}
Constructing $t$-designs from special NMDS codes has been an interesting research topic for a long time.
The first NMDS code dates back to 1949. Golay discovered the $[11,6,5]$ ternary NMDS code which is called the  ternary Golay code.
This NMDS code holds $4$-designs.
In the past 70 years after this discovery, only sporadic NMDS codes holding $t$-designs were found.
The question as to whether there exists an infinite family of NMDS codes holding an
infinite family of $t$-designs for $t\geq 2$ remained  open during this long period.
In 2020, Ding and Tang made a breakthrough in constructing the first two infinite families of NMDS codes holding
$2$-designs or $3$-designs \cite{T2}. Up to now, there are only a few known infinite families of NMDS codes holding $t$-designs for
$t=2,3,4$ in the literature \cite{T2, T7, 32D, YZ}. It is challenging to construct new infinite families of NMDS codes holding $t$-designs with $t>1$.

The objective of this paper is to construct several new infinite families of NMDS codes holding $t$-designs.
To this end, some special matrices over finite fields are used as the generator matrices of the NMDS codes.
We then determine the weight enumerators of the NMDS codes and prove that the NMDS codes hold $2$-deigns or $3$-designs.
Most of the NMDS codes in this paper don't satisfy the conditions in the Assmus-Mattson Theorem but still hold $t$-designs.
Compared with known $t$-designs from NMDS codes, ours have different parameters.
Besides, several infinite families of optimal locally recoverable codes are also derived via the NMDS codes.

\section{Preliminaries}

In this section, we present some preliminaries on the properties of NMDS codes and oval polynomials, and the number of zeros of some equations over finite fields.

\subsection{Properties of NMDS codes}
Let $\cC$ be an $[n,\kappa]$ linear code and $\cC^\perp$ be its dual.
Denote by $(1,A_1,\ldots,A_n)$ and $(1,A_1^\perp,\ldots,A_n^\perp)$ the weight distributions of $\cC$ and $\cC^\perp$, respectively.
If $\cC$ is an NMDS code, then the weight distributions of $\cC$ and $\cC^\perp$ are given in the following lemma.

\begin{lemma}\label{lem-nmdsweight}\cite{D}
Let $\mathcal{C}$ be an $[n, \kappa, n-\kappa]$ NMDS code over $\gf_q$. If $s \in \{1,2, \ldots, n-\kappa\}$, then
\begin{eqnarray*}
A_{\kappa+s}^\perp = \binom{n}{\kappa+s} \sum_{j=0}^{s-1} (-1)^j \binom{\kappa+s}{j}(q^{s-j}-1) +
             (-1)^s \binom{n-\kappa}{s}A_{\kappa}^\perp.
\end{eqnarray*}
If $s \in \{1,2, \ldots, \kappa\}$, then
\begin{eqnarray*}
A_{n-\kappa+s} = \binom{n}{\kappa-s} \sum_{j=0}^{s-1} (-1)^j \binom{n-\kappa+s}{j}(q^{s-j}-1)
+(-1)^s \binom{\kappa}{s}A_{n-\kappa}.
\end{eqnarray*}
\end{lemma}

Though Lemma \ref{lem-nmdsweight} and the relation $1+\sum_{s=0}^{\kappa}A_{n-\kappa+s}=q^{\kappa}$ hold, the weight distribution of an $[n,\kappa]$
NMDS code still can not be totally determined. In \cite{NMDS1}, some infinite families of NMDS codes with the same parameters but different weight distributions were constructed.

The following lemma establishes an interesting relationship between the minimum weight codewords in $\cC$ and those in $\cC^\perp$.
\begin{lemma}\label{lem-minweight}\cite{Lemma3}
Let $\mathcal{C}$ be an NMDS code. For any $\bc=(c_1, \ldots, c_{n})\in \C$, its support is defined by $\support(\bc)=\{1 \leq i \leq n: c_i \neq 0\}$.
Then for any minimum weight codeword $\bc$ in $\mathcal{C}$, there exists, up to a multiple, a unique minimum weight codeword $\bc^\perp$ in $\mathcal{C}^\perp$ satisfying
$\support(\bc) \cap \support(\bc^\perp)=\emptyset$. Besides, the number of  minimum weight codewords in $\mathcal{C}$ and the number of those in $\cC^\perp$ are the same.
\end{lemma}

By Lemma \ref{lem-minweight}, if the minimum weight codewords of an NMDS code hold a $t$-design, then the minimum weight codes of its dual hold a complementary $t$-design.

\subsection{Oval polynomials and their properties}
The definition of oval polynomial is presented as follows.

\begin{definition}\label{defn}  \cite{OVAL}
Let $q=2^m$ with $m \geq 2$. If $f \in \gf_q[x]$ is a polynomial such that $f$ is a permutation polynomial of $\gf_q$ with $\deg(f)<q$ and $f(0)=0$, $f(1)=1$, and
 $g_a(x):=(f(x+a)+f(a))x^{q-2}$ is also a permutation polynomial of $\gf_q$ for each $a \in \gf_q$, then $f$ is called an oval polynomial.
\end{definition}

The following gives some known oval polynomials.

\begin{lemma}\label{thm-knownopolys}\cite{M}
Let $m \geq 2$ be an integer. Then the followings are oval polynomials of $\gf_q$, where $q=2^m$.
\begin{enumerate}
\item The translation polynomial $f(x)=x^{2^h}$, where $\gcd(h, m)=1$.
\item The Segre polynomial $f(x)=x^6$, where $m$ is odd.
\end{enumerate}
\end{lemma}

By Lemma \ref{thm-knownopolys}, it is obvious that $f(x)=x^4$ is an oval polynomial of $\gf_q$, where $q=2^m$ for odd $m$.
By Definition \ref{defn}, the following also holds for oval polynomials.

\begin{lemma}\label{lem-oval}
Let $q=2^m$ with $m\geq 2$. Then $f$ is an oval polynomial of $\gf_q$ if and only if the followings simultaneously hold:
\begin{enumerate}
\item $f$ is a permutation polynomial of $\gf_q$ with $\deg(f)<q$ and $f(0)=0$, $f(1)=1$; and
\item
$$
\frac{f(x)+f(y)}{x+y} \neq \frac{f(x)+f(z)}{x+z}
$$
for all pairwise distinct elements $x, y, z$ in $\gf_q$.
\end{enumerate}
\end{lemma}

\subsection{The number of zeros of some equations over finite fields}
 Let $q=p^m$ with $p$ a prime. The following lemma is very useful for determining the greatest common divisor of some special integers.
\begin{lemma} \label{lem-num1}
Let $h$ and $m$ be two integers with $\gcd(h,m)= \ell $. Then
\begin{eqnarray*}
\gcd(p^h+1,q-1)=\left\{
\begin{array}{ll}
1, & \mbox{for odd $\frac{m}{\ell}$ and $p=2$},\\
2, & \mbox{for odd $\frac{m}{\ell}$ and odd $p$},\\
p^\ell+1, & \mbox{for even $\frac{m}{\ell}$}.
\end{array} \right.
\end{eqnarray*}
\end{lemma}

The following lemmas present some results on the number of solutions of some equations over $\gf_q$.
\begin{lemma}\label{lem-s}\cite{32D}
Let $n$ be a positive integer and $q$ a prime power such that $\gcd(n , q-1)=s $. Then $x^n-1$ has $s$ zeros in $\gf_q$.
\end{lemma}

\begin{lemma}\label{lem-ph1}\cite[Proof of Lemma 4]{32D}
Let $h$ be a positive integer. Denote by $N_g$ the number of zeros of $ g(x)=x^{p^h+1} + c$ in $\gf_q$, where $c \in \gf_q^*$. Then $N_g=0$ if and only if $-c$ is not a $(p^h+1)$-th power in $\gf_q^*$. If $-c$ is a $(p^h+1)$-th power in $\gf_q^*$, then $N_g=\gcd(p^h+1,q-1)$.
\end{lemma}

\begin{lemma}\label{lem-ker}\cite{OVAL}
The trace function from $\gf_q$ to $\gf_p$ is defined by
 $$\tr_{q/p}(x)=x+x^{p}+x^{p^2}+\cdots+x^{p^{m-1}}.$$
 Then for $\alpha \in \gf_q$, $\tr_{q/p}(\alpha)=0$ if and only if $\alpha = \beta^p-\beta$ for some $\beta \in \gf_q$.
\end{lemma}

\begin{lemma}\label{lem-abc}\cite{P}
Let $\gf_q$ be a finite field of characteristic 2 and let $f(x)=ax^2+bx+c \in \gf_q[x]$ be a polynomial of degree 2. Then
\begin{enumerate}
\item $f$ has exactly one root in $\gf_q$ if and only if $ b=0 $;
\item $f$ has exactly two roots in $\gf_q$ if and only if $ b \neq 0 $ and $\tr_{q/2}(\frac{ac}{b^2})=0$;
\item $f$ has no root in $\gf_q$ if and only if $ b \neq 0 $ and $\tr_{q/2}(\frac{ac}{b^2})=1$.
\end{enumerate}
\end{lemma}

Let $m$ and $h$ be positive integers and $q=2^m$. Now we consider the zeros of polynomials
\begin{eqnarray*}
f(x)=ax^{2^h+1}+bx+c,\ a,b,c \in \gf_q
\end{eqnarray*}
and
\begin{eqnarray*}
g(x)=ax^{2^h+1}+bx^{2^h}+c,\ a,b,c \in \gf_q
\end{eqnarray*}
in $\gf_q$. If $a,b,c \in \gf_q^*$, then $f(x)$ can be reduced to
\begin{eqnarray*}
P_\gamma(x)=x^{2^h+1}+x+\gamma,
\end{eqnarray*}
by the substitution $x\longmapsto ux$, where $u^{2^h}= \frac{b}{a}$ and $\gamma=\frac{c}{au^{2^h+1}}\in \gf_q^*$.
If $a,b,c \in \gf_q^*$, then $g(x)$ can be reduced to
\begin{eqnarray*}
U_\ell(x)=x^{2^h+1}+x^{2^h}+\ell,
\end{eqnarray*}
by the substitution $x\longmapsto vx$, where $v= \frac{b}{a}$ and $\ell=\frac{c}{av^{2^h+1}}\in \gf_q^*$.

\begin{lemma}\label{lem-013}\cite{MAIN}
Let $N_\gamma$ denote the number of zeros of $P_\gamma(x)$ in $\gf_q$ with $\gamma\in \gf_q^*$. If $\gcd(h , m)=1$, then $N_\gamma$ is equal to $0, 1$ or $3$.
\end{lemma}

\begin{lemma}\label{lem-f013}
Let $h$ and $m$ be positive integers with $\gcd(h , m)=1 $ and $q=2^m$. Denote by $N_f$ the number of zeros of $f(x)=ax^{2^h+1}+bx+c$ in $\gf_q^*$, where $(a,b,c) \neq (0,0,0), a,b,c \in \gf_q$. Then $N_f\in \{0,1,3\}$.
\end{lemma}

\begin{proof}
It is obvious that $N_f$ is equal to 0 or 1 if $a=0$ or $b=c=0$. Now let $a\neq0$.
 If $b \neq 0$ and $c=0$, then $f(x)=ax^{2^h+1}+bx=ax(x^{2^h}+\frac{b}{a})$. Since $q$ is even,  it is clear that $f(x)$ has only one zero in $\gf_q^*$. If $b=0$ and $c \neq 0$, then $f(x)=ax^{2^h+1}+c=a(x^{2^h+1}+\frac{c}{a})$ and $N_f\in \{0,1,3\}$ by Lemmas \ref{lem-num1} and \ref{lem-ph1}. If  $b \neq 0$ and $c \neq 0$, by Lemma \ref{lem-013}, $N_f=N_\gamma\in \{0,1,3\}$. Then the desired conclusion follows.
\end{proof}

\begin{lemma}\label{lem-013plus}
Let $N_\ell$ denote the number of zeros of $U_\ell(x)$ in $\gf_q^*$ with $\ell\in \gf_q^*$. If $\gcd(h , m)=1$, then $N_\ell$ is equal to $0, 1$ or $3$.
\end{lemma}

\begin{proof}
Let $x_0$ be a zero of $U_\ell(x)$ in $\gf_q^*$, then it is easy to prove that $P_\ell(x_0+1)=U_\ell(x_0)=0$.
Let $x_0$ be a zero of $P_\gamma(x)$ in $\gf_q^*$, then we have $U_\gamma(x_0+1)=P_\gamma(x_0)=0$.
Then it is easy to deduce $N_\ell=N_\gamma$. The desired conclusion follows from Lemma \ref{lem-013}.
\end{proof}

\begin{lemma}\label{lem-g013}
Let $h$ and $m$ be positive integers with $\gcd(h , m)=1 $ and $q=2^m$. Denote by $N_g$ the number of zeros of $g(x)=ax^{2^h+1}+bx^{2^h}+c$ in $\gf_q^*$, where $(a,b,c) \neq (0,0,0), a,b,c \in \gf_q$. Then $N_g\in \{0,1,3\}$.
\end{lemma}

\begin{proof}
Similarly to the proof of Lemma \ref{lem-f013}, the desired conclusion follows from Lemmas \ref{lem-num1}, \ref{lem-ph1} and \ref{lem-013plus}.
\end{proof}

\begin{lemma}\label{lem-ab}
Let $q=2^m$, where $m$ is an odd integer with $m \geq 3$. Then for two distinct elements $a,b\in \gf_q$,  the polynomial $u(x)=x^2+(a+b)x+a^2+b^2+ab$ has no root in $\gf_q$. In other words, for any $c\in \gf_q$, we have $a^2+b^2+c^2+ab+ac+bc \neq 0$. Particularly, if $c=0$, then $a^2+b^2+ab \neq 0$.
\end{lemma}

\begin{proof}
It is obvious that
$$\tr_{q/2}\left(\frac{a^2+b^2+ab}{(a+b)^2}\right) = \tr_{q/2}(1) + \tr_{q/2}\left(\frac{ab}{a^2+b^2}\right). $$
Let $\beta = \frac{a}{a+b} \in \gf_q$. It is clear that $\frac{ab}{a^2+b^2}=\beta^2-\beta$. By Lemma \ref{lem-ker},
$$ \tr_{q/2}\left(\frac{ab}{a^2+b^2}\right) =0.$$
When $m$ is odd, $\tr_{q/2}\left(\frac{a^2+b^2+ab}{(a+b)^2}\right) = \tr_{q/2}(1)=1$. Then by Lemma \ref{lem-abc}, $u(x)$ has no root in $\gf_q$. The proof is completed.
\end{proof}

\subsection{Generalized Vandermonde determinant}
The following  provides a general equation for a generalized Vandermonde determinant with one deleted row in terms of the elementary symmetric polynomial.

\begin{lemma}\label{lem-Vandermonde}\cite[Lemma 17]{LDMT}\cite[Page 466]{ERH}
For each $\ell$ with $0 \leq \ell \leq n$, it holds that
\begin{eqnarray*}
\left|
\begin{array}{ccccc}
1 & 1 & \cdots & 1  & 1\\
u_1 & u_2 & \cdots & u_{n-1} & u_{n}\\
 \vdots   & \vdots   & \vdots                &  \vdots     & \vdots          \\
u_1^{\ell-1} & u_2^{\ell-1} & \cdots & u_{n-1}^{\ell-1} & u_{n}^{\ell-1}\\
u_1^{\ell+1} & u_2^{\ell+1} & \cdots & u_{n-1}^{\ell+1} & u_{n}^{\ell+1}\\
 \vdots   & \vdots   & \vdots                &  \vdots     & \vdots          \\
u_1^n & u_2^n & \cdots & u_{n-1}^n & u_{n}^n
\end{array}
\right|
=\left(\prod_{1 \leqslant i<j \leqslant n}(u_j-u_i)\right)\sigma_{n-\ell}(u_1, u_2, \ldots, u_n),
\end{eqnarray*}
where
$$\sigma_{n-\ell}(u_1, u_2, \ldots, u_n)= \sum_{1 \leqslant i_1< \ldots <i_{n-\ell} \leqslant n}u_{i_1}\cdots u_{i_{n-\ell}}$$
is the $(n-\ell)$-th elementary symmetric polynomial over the set $\{u_1, u_2, \ldots, u_n\}$.
\end{lemma}



\section{Two families of 3-dimensional near MDS codes holding 2-designs}\label{sect-code1}

In this section, let $q=2^m$ with $m \geq 3$. Hereafter, let $\dim(\cC)$ and $d(\cC)$ respectively denote the dimension and minimum distance of a linear code $\cC$.  Let $\alpha$ be a generator of $\gf_q^*$ and $\alpha_{i}:=\alpha^{i}$ for $1 \leq i \leq q-1$. Then $\alpha_{q-1}=1$.

Let $h$ be a positive integer with $\gcd(m,h)=1$. Define
\begin{eqnarray*}\label{eqn-construction1}
D=\left[
\begin{array}{lllll}
1 & 1 & \cdots & 1  & 1\\
\alpha_1 & \alpha_2 & \cdots & \alpha_{q-2} & \alpha_{q-1}\\
\alpha_1^{2^h+1} & \alpha_2^{2^h+1} & \cdots & \alpha_{q-2}^{2^h+1} & \alpha_{q-1}^{2^h+1}
\end{array}
\right].
\end{eqnarray*}
 $D$ is a $3$ by $q-1$ matrix over $\gf_q$. Let $\cC_D$ be the linear code over $\gf_q$ generated by $D$.
We will show that $\cC_D$ is an NMDS code and both $\cC_D$ and its dual $\cC_D^\perp$ support 2-designs.

\begin{theorem}\label{thm-0}
Let $q=2^m$ with $m \geq 3$, $h$ be a positive integer with $\gcd(m,h)=1$. Then $\cC_D$ is a $[q-1, 3, q-4]$ NMDS code over $\gf_q$ with weight enumerator
\begin{eqnarray*}
A(z)=1 + \frac{(q-1)^2(q-2)}{6} z^{q-4} + \frac{(q-1)^2(q+4)}{2} z^{q-2}+\frac{(q-1)(q^2+8)}{3} z^{q-1}.
\end{eqnarray*}
Moreover, the minimum weight codewords of $\cC_D$ support a 2-$(q-1,q-4,\frac{(q-4)(q-5)}{6})$ simple design and the minimum weight codewords of $\cC_D^\perp$ support a 2-$(q-1,3,1)$ simple design, i.e. a Steiner system $S(2,3,q-1)$. Furthermore, the codewords of weight 4 in $\cC_D^\perp$ support a 2-$(q-1, 4, \frac{(q-4)(q-7)}{2})$ simple design.
\end{theorem}

\begin{proof}
We first prove that $\dim(\cC_D)=3$. Let $\bg_1$, $\bg_2$ and $\bg_3$ respectively represent the first, second and third rows of $D$. Assume that there exist elements $a,b,c \in \gf_q$ with $(a,b,c)\neq (0,0,0)$ such that $c \bg_1 + b \bg_2 + a \bg_3=\mathbf{0}$. Then
\begin{eqnarray*}
\left\{
\begin{array}{c}
a\alpha_1^{2^h+1}+b\alpha_1+c=0, \\
a\alpha_2^{2^h+1}+b\alpha_2+c=0, \\
          \vdots                 \\
a\alpha_{q-1}^{2^h+1}+b\alpha_{q-1}+c=0.
\end{array}
\right.
\end{eqnarray*}
This contradicts with the fact that the polynomial $f(x)=ax^{2^h+1}+bx+c$ has at most 3 zeros in $\gf_q^*$ by Lemma \ref{lem-f013}.   Hence $\bg_1$, $\bg_2$ and $\bg_3$ are linearly independent over $\gf_q$ and $\dim(\cC_D)=3$.

We then prove that $\cC_D^\perp$ has parameters $[q-1, q-4, 3]$. Obviously, $\dim(\cC_D^\perp)=(q-1)-3=q-4$. It is clear that each column of $D$ is nonzero and any two columns of $D$ are linearly independent over $\gf_q$. Then $d(\cC_D^\perp) > 2$. Let $x_1, x_2, x_3$ be three pairwise different elements in $\gf_q^*$. Consider the following submatrix as
\begin{eqnarray*}
D_{1}=\left[
\begin{array}{lll}
1 & 1 & 1 \\
x_1 & x_2 & x_3 \\
x_1^{2^h+1} & x_2^{2^h+1} & x_3^{2^h+1}
\end{array}
\right].
\end{eqnarray*}
Then $D$ has 3 columns that are linearly dependent if and only if $|D_{1}|=0$ for some $(x_1,x_2,x_3)$.  Besides, $\rank(D_{1})=2$ if $|D_{1}|=0$. Now we consider the following two cases.

\emph{Case 1:} Let $m$ be even. By Lemmas \ref{lem-num1} and \ref{lem-s}, the polynomial $x^{2^h+1}-1$ has 3 zeros in $\gf_q^*$ denoted by $r_1,r_2$ and $r_3$. Let $(x_1,x_2,x_3)=(r_1,r_2,r_3)$. Then
\begin{eqnarray*}
D_{1}=\left[
\begin{array}{lll}
1 & 1 & 1 \\
r_1 & r_2 & r_3 \\
1 & 1 & 1
\end{array}
\right].
\end{eqnarray*}
Thus $|D_{1}|=0$.

\emph{Case 2:} Let $m$ be odd and $x_3=\alpha_{q-1}=1$. Then
\begin{eqnarray*}
D_{1}=\left[
\begin{array}{lll}
1 & 1 & 1\\
x_1 & x_2 & 1\\
x_1^{2^h+1} & x_2^{2^h+1} & 1
\end{array}
\right].
\end{eqnarray*}
It is easy to deduce that $|D_{1}|=(1+x_1)x_2^{2^h+1}+(1+x_1^{2^h+1})x_2+x_1+x_1^{2^h+1}$. Denote by $f(x)=(1+x_1)x^{2^h+1}+(1+x_1^{2^h+1})x+x_1+x_1^{2^h+1}$.
Note that $1+x_1\neq 0$ as $x_1\neq x_3$. By Lemma \ref{lem-f013}, $f(x)$  has 0 or 1 or 3 zeros in $\gf_q^*$. It is easy to verify that $f(1)=f(x_1)=0$.
Then  exists  an element $r \in \gf_q^*$ which is different from 1 and $x_1$ such that $f(r)=0$. Let $x_2=r$ and we have $|D_{1}|=0$.

Summarizing  the above cases yields that $d(\cC_D^\perp) = 3$. Therefore, $\cC_D^\perp$ has parameters $[q-1, q-4, 3]$.

By definition, we have
$$ \cC_D= \{ \bc_{a,b,c}=(ax^{2^h+1}+bx+c)_{x \in \gf_q^*},a,b,c \in \gf_q \}. $$
To determine the weight $\mbox{wt}(\bc_{a,b,c})$ of a codeword $\bc_{a,b,c}\in \cC_D$, it is sufficient to determine the number of zeros of the equation
$$ ax^{2^h+1}+bx+c =0$$
in $\gf_q^*$.
By Lemma \ref{lem-f013}, the above equation has 0 or 1 or 3 zeros in $\gf_q^*$. Hence, $\mbox{wt}(\bc_{a,b,c})\in \{q-1,q-2,q-4\}$.

Finally, we compute the weight enumerator of $\cC_D$ by the first three Pless Power Moments in \cite{HP} and prove that $\cC_D$ is a $[q-1, 3, q-4]$ NMDS code. Let $A_{w_1},A_{w_2},A_{w_3}$ respectively represent the frequencies of the weights $w_1=q-4, w_2=q-2, w_3=q-1$. Then we have
\begin{eqnarray*}
\left\{
\begin{array}{ll}
A_{w_1}+A_{w_2}+A_{w_3}=q^3-1, \\
w_1A_{w_1}+w_2A_{w_2}+w_3A_{w_3}=q^2(q-1)^2, \\
w_1^2A_{w_1}+w_2^2A_{w_2}+w_3^2A_{w_3}=q(q-1)^2(q^2-2q+2).
\end{array}
\right.
\end{eqnarray*}
 Solving the above system of linear equations gives
$$ A_{q-4}= \frac{(q-1)^2(q-2)}{6}, A_{q-2}=\frac{(q-1)^2(q+4)}{2}, A_{q-1}=\frac{(q-1)(q^2+8)}{3}. $$
Thus $\cC_D$ is a $[q-1, 3, q-4]$ NMDS code and the weight enumerator of $\cC_D$ follows from Lemma \ref{lem-nmdsweight}. It then follows from  the Assmus-Mattson Theorem in Theorem \ref{the-AM} and Equation (\ref{eqn-t}) that the minimum weight codewords of $\cC_D$ support a 2-$(q-1,q-4,\frac{(q-4)(q-5)}{6})$ simple design, and the minimum weight codewords of $\cC_D^\perp$ support a 2-$(q-1,3,1)$ simple design.

We finally prove that the codewords of weight 4 in $\cC_D^\perp$ support a 2-$(q-1, 4, \frac{(q-4)(q-7)}{2})$ simple design. Thanks to a generalized version of the Assmus-Mattson Theorem (Theorem 2.2 in \cite {TDX}), the codewords of weight 4 in $\cC_D^\perp$ support a 2-design. We need to prove that this design is simple. Let $x,y,z$ be three pairwise distinct elements in $\gf_q^*$. Define
\begin{eqnarray*}
D_2=
\left[
\begin{array}{lll}
 1 & 1  & 1\\
 x & y & z\\
x^{2^h+1} & y^{2^h+1} & z^{2^h+1}
\end{array}
\right].
\end{eqnarray*}
It is obvious that
$$  |D_2| =z^{2^h+1}(x+y)+z(x^{2^h+1}+y^{2^h+1})+xy^{2^h+1}+yx^{2^h+1}.  $$
Let $f(z)=z^{2^h+1}(x+y)+z(x^{2^h+1}+y^{2^h+1})+xy^{2^h+1}+yx^{2^h+1}$. Note that $f(x)=f(y)=0.$ By Lemma \ref{lem-f013}, $f(z)$ has 0 or 1 or 3 zeros in $\gf_q^*$. Thus there exists an element $r_{(x,y)} \in \gf_q^*$ which is different from $x$ and $y$ such that $f(r_{(x,y)})=0$. Then we have rank$(D_2) = 3$ if and only if $z \notin \{ x, y, r_{(x,y)} \}$.
Next we prove that the rank of the submatrix
\begin{eqnarray*}
D_{(x,y,z,w)}=
\left[
\begin{array}{llll}
 1 & 1 & 1 & 1 \\
 x & y & z & w \\
x^{2^h+1} & y^{2^h+1} & z^{2^h+1} & w^{2^h+1}
\end{array}
\right]
\end{eqnarray*}
equals 3 for any four pairwise distinct elements $x,y,z,w \in \gf_q^*$. It is obvious that at least one of $z, w$ is not equal to $r_{(x,y)}$, which implies $D_{(x,y,z,w)}$ has a three-order non-zero minor. Thus $\rank(D_{(x,y,z,w)})=\rank(D_2)=3$. Let $\bc=(c_1,c_2, \ldots, c_{q-1})$ be a codeword of weight 4 in $\cC_D^\perp$ with nonzero coordinates in $\{ i_1, i_2, i_3, i_4 \}$, which means $c_{i_j} \ne 0 $ for $1 \leq j \leq 4$ and $c_v=0$ for all $v \in \{1,2, \ldots, q-1\} \setminus \{i_1, i_2, i_3, i_4\}$. Since $ D $ is a parity-check matrix of $\cC_D^\perp$,  there exist four pairwise distinct elements $x,y,z,w \in \gf_q^*$ such that
\begin{eqnarray*}
\left[
\begin{array}{llll}
 1 & 1 & 1 & 1 \\
 x & y & z & w \\
x^{2^h+1} & y^{2^h+1} & z^{2^h+1} & w^{2^h+1}
\end{array}
\right]
\left[
\begin{array}{l}
c_{i_1}\\
c_{i_2}\\
c_{i_3}\\
c_{i_4}
\end{array}
\right]
=\mathbf{0}.
\end{eqnarray*}
Since rank$(D_{(x,y,z,w)})=3$, the all nonzero solutions of the above equation are $\{a(c_{i_1},c_{i_2},c_{i_3},c_{i_4}):a \in \gf_q^* \}$. Thus $\{a\bc:a \in \gf_q^* \}$ is a set of all codewords of weight 4 in $\cC_D^\perp$ whose nonzero coordinates are $\{i_1, i_2, i_3, i_4 \}$. Hence, the codewords of weight 4 in $\cC_D^\perp$ support a 2-$(q-1, 4, \lambda )$ simple design. Since $\cC_D^\perp$ is an NMDS code, we have $A_4^\perp = \frac{(q-1)^2(q-2)(q-4)(q-7)}{24}$  by Lemma \ref{lem-nmdsweight}. By Equation (\ref{eqn-t}), we have
$\lambda = \frac{(q-4)(q-7)}{2}$.

Then we have completed the proof.
\end{proof}

Let $h$ be a positive integer with $\gcd(m,h)=1$. Define
\begin{eqnarray*}\label{eqn-construction1}
H=\left[
\begin{array}{lllll}
1 & 1 & \cdots & 1  & 1\\
\alpha_1^{2^h} & \alpha_2^{2^h} & \cdots & \alpha_{q-2}^{2^h} & \alpha_{q-1}^{2^h}\\
\alpha_1^{2^h+1} & \alpha_2^{2^h+1} & \cdots & \alpha_{q-2}^{2^h+1} & \alpha_{q-1}^{2^h+1}
\end{array}
\right].
\end{eqnarray*}
 $H$ is a $3$ by $q-1$ matrix over $\gf_q$. Let $\cC_H$ be the linear code over $\gf_q$ generated by $H$.
We will show that $\cC_H$ is an NMDS code and both $\cC_H$ and its dual $\cC_H^\perp$ support 2-designs.

\begin{theorem}\label{thm-0+}
Let $q=2^m$ with $m \geq 3$, $h$ be a positive integer with $\gcd(m,h)=1$. Then $\cC_H$ is a $[q-1, 3, q-4]$ NMDS code over $\gf_q$ with weight enumerator
\begin{eqnarray*}
A(z)=1 + \frac{(q-1)^2(q-2)}{6} z^{q-4} + \frac{(q-1)^2(q+4)}{2} z^{q-2}+\frac{(q-1)(q^2+8)}{3} z^{q-1}.
\end{eqnarray*}
Moreover, the minimum weight codewords of $\cC_H$ support a 2-$(q-1,q-4,\frac{(q-4)(q-5)}{6})$ simple design and the minimum weight codewords of $\cC_H^\perp$ support a 2-$(q-1,3,1)$ simple design, i.e. a Steiner system $S(2,3,q-1)$. Furthermore, the codewords of weight 4 in $\cC_H^\perp$ support a 2-$(q-1, 4, \frac{(q-4)(q-7)}{2})$ simple design.
\end{theorem}

\begin{proof}
Similarly to the proof of Theorem \ref{thm-0}, we can easily derive this theorem by Equation (\ref{eqn-t}), Lemmas \ref{lem-nmdsweight} and \ref{lem-g013}.
\end{proof}

Note that $\cC_D$ and $\cC_H$ have the same parameters and weight enumerator.
It is open whether they are equivalent to each other.

\section{Two families  of 4-dimensional near MDS codes holding 2-designs}\label{sect-code2}

In this section, let $q=2^m$ with $m \geq 3$.  Let $\alpha$ be a generator of $\gf_q^*$ and $\alpha_{i}:=\alpha^{i}$ for $1 \leq i \leq q-1$. Then $\alpha_{q-1}=1$.
Define a $4$ by $q-1$ matrix over $\gf_q$ by
\begin{eqnarray*}\label{eqn-construction2and3}
G_{(i,j)}=\left[
\begin{array}{lllll}
1 & 1 & \cdots & 1  & 1\\
\alpha_1^i & \alpha_2^i & \cdots & \alpha_{q-2}^i & \alpha_{q-1}^i\\
\alpha_1^j & \alpha_2^j & \cdots & \alpha_{q-2}^j & \alpha_{q-1}^j\\
\alpha_1^4 & \alpha_2^4 & \cdots & \alpha_{q-2}^4 & \alpha_{q-1}^4
\end{array}
\right],
\end{eqnarray*}
where $(i,j)=(1,3)$ or $(2,3)$.  Let $\cC_{(i,j)}$ be the linear code over $\gf_q$ generated by $G_{(i,j)}$.
We will prove that $\cC_{(i,j)}$ is an NMDS code and the minimum weight codewords of $\cC_{(i,j)}$ and its dual $\cC_{(i,j)}^\perp$ support 2-designs.

\subsection{When $(i,j)=(1,3)$}

The following lemma plays an important role in the proof of our main result.
\begin{lemma}\label{thm-prooft1}
Let $m$ be an odd integer with $m>3$, $q=2^m$. Let $x_1, x_2, x_3, x_4$ be four pairwise distinct elements in $\gf_q^*$ and we define the matrix
\begin{eqnarray}\label{eqn-m1}
M_{(1,3)}=
\left[
\begin{array}{llll}
1 & 1 & 1  & 1\\
x_1 & x_2 & x_3 & x_4\\
x_1^3 & x_2^3 & x_3^3 & x_4^3\\
x_1^4 & x_2^4 & x_3^4 & x_4^4
\end{array}\right].
\end{eqnarray}
Then for any two different and fixed  elements $x_1, x_2$, the total number of different choices of $x_3, x_4$ such that $|M_{(1,3)}|=0$ is equal to $\frac{q-8}{2}$ (regardless of the ordering of $x_3, x_4$).
\end{lemma}

\begin{proof}
By Lemma \ref{lem-Vandermonde}, $|M_{(1,3)}|=0$ if and only if $x_1x_2+x_1x_3+x_2x_3+(x_1+x_2+x_3)x_4=0$. Then we first need to consider whether $x_1+x_2+x_3$ equals 0 or not in the following cases.

\emph{Case 1:} Let $x_1+x_2+x_3=0$. Then $x_3=x_1+x_2$ and
$$|M_{(1,3)}|=\prod_{1 \leqslant i<j \leqslant 4}(x_j-x_i)(x_1^2+x_2^2+x_1x_2) \neq 0$$
by Lemma \ref{lem-ab}.
So there is no  $(x_3, x_4)$ such that $|M_{(1,3)}|=0$ in this case.

\emph{Case 2:} Let $x_1+x_2+x_3 \neq 0$. Then $x_3 \neq x_1+x_2$ and
$$|M_{(1,3)}|=0 \Longleftrightarrow x_1x_2+x_1x_3+x_2x_3+(x_1+x_2+x_3)x_4=0 \Leftrightarrow x_4=\frac{x_1x_2+x_1x_3+x_2x_3}{x_1+x_2+x_3}. $$
Since $x_1, x_2, x_3, x_4$ are four pairwise distinct elements in $\gf_q^*$, we have $x_4 \notin \{0,x_1,x_2,x_3\}$. Note that
 $$x_4 \neq 0\Longleftrightarrow \frac{x_1x_2+x_1x_3+x_2x_3}{x_1+x_2+x_3} \neq 0 \Leftrightarrow x_1x_2+x_1x_3+x_2x_3 \neq 0 \Leftrightarrow x_3 \neq \frac{x_1x_2}{x_1+x_2}.$$
 Similarly,  $x_4 \notin \{x_1,x_2,x_3\}$ if and only if $x_3 \notin \{\frac{x_2^2}{x_1},\frac{x_1^2}{x_2},a \}$,where $a^2=x_1x_2$.
 We conclude that $|M_{(1,3)}|=0$ if and only if
\begin{eqnarray*}
x_3 \notin \left\{0, x_1, x_2, x_1+x_2, \frac{x_1x_2}{x_1+x_2}, \frac{x_2^2}{x_1}, \frac{x_1^2}{x_2}, a \right\},
\end{eqnarray*}
where $a^2=x_1x_2$, and
$$x_4=\frac{x_1x_2+x_1x_3+x_2x_3}{x_1+x_2+x_3}.$$
By Lemmas \ref{lem-s} and \ref{lem-ab}, it is easy to prove that the elements in
$$ \left\{0, x_1, x_2, x_1+x_2, \frac{x_1x_2}{x_1+x_2}, \frac{x_2^2}{x_1}, \frac{x_1^2}{x_2}, a \right\} $$
are pairwise distinct. It is obvious that if $(x_3,x_4)$ is a choice, so is $(x_4,x_3)$. Hence the total number of different choices of $x_3, x_4\in \gf_q^*$ such that $|M_{(1,3)}|=0$ is equal to $\frac{q-8}{2}$ for any two different fixed  elements $x_1, x_2$.

The proof is completed.
\end{proof}

\begin{theorem}\label{thm-1}
Let $m$ be an odd integer with $m>3$, $q=2^m$. Then $\cC_{(1,3)}$ is a $[q-1, 4, q-5]$ NMDS code over $\gf_q$ with weight enumerator
\begin{eqnarray*}
A(z)=1 + \frac{(q-1)^2(q-2)(q-8)}{24} z^{q-5} + \frac{5(q-1)^2(q-2)}{6} z^{q-4}+\frac{q(q-1)^2(q-2)}{4} z^{q-3} \\ +\frac{(q-1)^2(2q^2+7q+20)}{6} z^{q-2}+\frac{(q-1)(9q^3+13q^2-6q+80)}{24} z^{q-1}.
\end{eqnarray*}
Moreover, the minimum weight codewords of $\cC_{(1,3)}$ support a 2-$(q-1,q-5,\frac{(q-5)(q-6)(q-8)}{24})$ simple design and the minimum weight codewords of $\cC_{(1,3)}^\perp$ support a 2-$(q-1,4,\frac{q-8}{2})$ simple design.
\end{theorem}

\begin{proof}
We first prove that $\dim(\cC_{(1,3)})=4$. Let $\bg_1$, $\bg_2$, $\bg_3$ and $\bg_4$ respectively represent the first, second, third and fourth rows of $G_{(1,3)}$. Assume that there exist elements $a,b,c,d \in \gf_q$ with $(a,b,c,d)\neq (0,0,0,0)$ such that $a \bg_1 + b \bg_2 + c \bg_3+ d \bg_4=\mathbf{0}$. Then
\begin{eqnarray*}
\left\{
\begin{array}{c}
a+b\alpha_1+c\alpha_1^{3}+d\alpha_1^{4}=0, \\
a+b\alpha_2+c\alpha_2^{3}+d\alpha_2^{4}=0, \\
          \vdots                 \\
a+b\alpha_{q-1}+c\alpha_{q-1}^{3}+d\alpha_{q-1}^{4}=0.
\end{array}
\right.
\end{eqnarray*}
Obviously, the polynomial $f(x)=a+bx+cx^{3}+dx^{4}$ has at most 4 zeros in $\gf_q^*$, which leads to a contradiction. Hence $\bg_1$, $\bg_2$, $\bg_3$ and $\bg_4$ are linearly independent over $\gf_q$. Thus $\dim(\cC_{(1,3)})=4$.

We then prove that $\cC_{(1,3)}^\perp$ has parameters $[q-1, q-5, 4]$. Obviously, $\dim(\cC_{(1,3)}^\perp)=(q-1)-4=q-5$. We need to prove $d(\cC_{(1,3)}^\perp)=4$. It is sufficient to prove that any 3 columns of $G_{(1,3)}$ are linearly independent and there exist 4 columns of $G_{(1,3)}$ that are linearly dependent. Choosing  any three columns from $G_{(1,3)}$ yields the submatrix
\begin{eqnarray*}
M_{1.1}=
\left[
\begin{array}{lll}
1 & 1 & 1  \\
x_1 & x_2 & x_3 \\
x_1^3 & x_2^3 & x_3^3 \\
x_1^4 & x_2^4 & x_3^4
\end{array}\right],
\end{eqnarray*}
where $x_1,x_2,x_3$ are pairwise distinct elements in $\gf_q^*$. Consider  the 3 by 3 submatrix of $M_{1.1}$ as
\begin{eqnarray*}
M_{1.2}=
\left[
\begin{array}{lll}
1 & 1 & 1  \\
x_1 & x_2 & x_3 \\
x_1^4 & x_2^4 & x_3^4
\end{array}\right].
\end{eqnarray*}
 Denote by $f(x)=x^4$. By Lemma \ref{lem-oval}, we have $|M_{1.2}|=(x_2+x_1)(f(x_3)+f(x_1))+(x_3+x_1)(f(x_2)+f(x_1)) \neq 0$. Then $\rank(M_{1.2})=3$ and any 3 columns of $G_{(1,3)}$ are linearly independent.
Now we consider the submatrix $M_{(1,3)}$ of $G_{(1,3)}$ in Equation (\ref{eqn-m1}). By the proof of Lemma \ref{thm-prooft1}, there is proper $\{x_1, x_2, x_3, x_4\}$ such that $|M(1,3)| = 0.$ This shows that there exist 4 columns of $G(1,3)$ that are linearly dependent. To sum up, $d(\cC_{(1,3)}^\perp)=4$. Let $\bc=(c_1,c_2, \ldots, c_{q-1}) \in \cC_{(1,3)}^\perp$ and wt$(\bc)=4$. Assume that $c_{i_j}=r_{i_j} \in \gf_q^*$, $1 \leq j \leq 4$, and $c_v=0$ for all $v \in \{1,2, \ldots, q-1\} \setminus \{i_1, i_2, i_3, i_4\}$, i.e., $\support(\bc)= \{i_1, i_2, i_3, i_4\}$. Set $x_j=\alpha^{i_j}.$ By definition,
\begin{eqnarray*}
\left[
\begin{array}{llll}
1 & 1 & 1  & 1\\
x_1 & x_2 & x_3 & x_4\\
x_1^3 & x_2^3 & x_3^3 & x_4^3\\
x_1^4 & x_2^4 & x_3^4 & x_4^4
\end{array}
\right]
\left[
\begin{array}{l}
r_{i_1}\\
r_{i_2}\\
r_{i_3}\\
r_{i_4}
\end{array}
\right]
=\mathbf{0}.
\end{eqnarray*}
Since rank$(M_{(1,3)})=3$, the number of solutions with $\{r_{i_1}, r_{i_2}, r_{i_3}, r_{i_4}\} \in (\gf_q^*)^4$ is $q-1$. Then $\{a\bc:a \in \gf_q^* \}$ is the set of all codewords of weight 4 in $\cC_{(1,3)}^\perp$ whose nonzero coordinates are $\{i_1, i_2, i_3, i_4 \}$. Therefore, every codeword of weight 4 and its nonzero multiples in $\cC_{(1,3)}^\perp$ with nonzero coordinates $\{i_1, i_2, i_3, i_4 \}$ must correspond to the set $\{x_1 , x_2, x_3, x_4\}$. By Lemma \ref{thm-prooft1}, the number of choices of $x_3,x_4$ is independent of $x_1,x_2$. We then deduce that the codewords of weight 4 in $\cC_{(1,3)}^\perp$ support a 2-$(q-1,4,\frac{q-8}{2})$ design. Then by Equation (\ref{eqn-t}),
$$ A_4^\perp=(q-1)\frac{\binom{q-1}{2}}{\binom{4}{2}}\frac{q-8}{2}=\frac{(q-1)^2(q-2)(q-8)}{24}.$$
We next prove $d(\cC_{(1,3)})=q-5$. By definition, we have
$$ \cC_{(1,3)}= \{ \bc_{a,b,c,d}=(a+bx+cx^{3}+dx^{4})_{x \in \gf_q^*},a,b,c,d \in \gf_q \}. $$
To determine the weight $\mbox{wt}(\bc_{a,b,c,d})$ of a codeword $\bc_{a,b,c,d}\in \cC_{(1,3)}$, it is sufficient to determine the number of zeros of the equation
$$ a+bx+cx^{3}+dx^{4} =0$$
in $\gf_q^*$.
The above equation has at most 4 zeros in $\gf_q^*$. Hence, $d(\cC_{(1,3)}) \geq q-5$. By the Singleton bound,
$d(\cC_{(1,3)}) \leq q-4$. Then $d(\cC_{(1,3)})=q-4$ or $q-5$. If $d(\cC_{(1,3)})=q-4$, then $\cC_{(1,3)}$ is an MDS code. Then $\cC_{(1,3)}^\perp$ is also MDS, which leads to a contradiction. Therefore, $\cC_{(1,3)}$ is a $[q-1,4,q-5]$ NMDS code. By Lemma \ref{lem-minweight},
$$A_{q-5}=A_4^\perp=\frac{(q-1)^2(q-2)(q-8)}{24}.$$
Then by  Lemma \ref{lem-minweight} and Equation (\ref{eqn-t}), the minimum weight codewords of $\cC_{(1,3)}$ support a 2-$(q-1,q-5,\frac{(q-5)(q-6)(q-8)}{24})$ simple design. Finally, the weight enumerator of $\cC_{(1,3)}$ follows from Lemma \ref{lem-nmdsweight}.
\end{proof}

\subsection{When $(i,j)=(2,3)$}
In this subsection, we consider the case for $(i,j)=(2,3)$.

We will need the following lemma in the proof of our main result.
\begin{lemma}\label{thm-prooft2}
Let $m$ be an integer with $m \geq 3$, $q=2^m$. Let $x_1, x_2, x_3, x_4$ be four pairwise distinct elements in $\gf_q^*$. Define the matrix
\begin{eqnarray*}\label{eqn-m2}
M_{(2,3)}=
\left[
\begin{array}{llll}
1 & 1 & 1  & 1\\
x_1^2 & x_2^2 & x_3^2 & x_4^2\\
x_1^3 & x_2^3 & x_3^3 & x_4^3\\
x_1^4 & x_2^4 & x_3^4 & x_4^4
\end{array}\right].
\end{eqnarray*}
Then for any different fixed elements $x_1, x_2$, the total number of different choices of $x_3, x_4$ such that $|M_{(2,3)}|=0$ is equal to $\frac{q-4}{2}$.
\end{lemma}

\begin{proof}
Similarly to the proof of Lemma \ref{thm-prooft1},
we can easily derive this lemma by Lemma \ref{lem-Vandermonde}.
\end{proof}

\begin{theorem}\label{thm-2}
Let $m$ be an integer with $m \geq 3$, $q=2^m$. Then $\cC_{(2,3)}$ is a $[q-1, 4, q-5]$ NMDS code over $\gf_q$ with weight enumerator
\begin{eqnarray*}
A(z)=1 + \frac{(q-1)^2(q-2)(q-4)}{24} z^{q-5} + \frac{(q-1)^2(q-2)}{6} z^{q-4}+\frac{(q-1)^2(q-2)(q+4)}{4} z^{q-3} \\ +\frac{(q-1)^2(2q^2+3q+28)}{6} z^{q-2}+\frac{(q-1)(9q^3+17q^2-18q+88)}{24} z^{q-1}.
\end{eqnarray*}
Moreover, the minimum weight codewords of $\cC_{(2,3)}$ support a 2-$(q-1,q-5,\frac{(q-4)(q-5)(q-6)}{24})$ simple design and the minimum weight codewords of $\cC_{(2,3)}^\perp$ support a 2-$(q-1,4,\frac{q-4}{2})$ simple design.
\end{theorem}

\begin{proof}
Similarly to the proof of Theorem \ref{thm-1}, we can prove this theorem by Lemma \ref{thm-prooft2}.
\end{proof}

Note that Theorem \ref{thm-2} works for any $m \geq 3$ as Lemma \ref{thm-prooft2} dose not rely on Lemma \ref{lem-ab}.
\section{Five families  of 5-dimensional near MDS codes holding 2-designs or 3-designs}\label{sect-code3}
In this section, let $q=2^m$ with $m > 3$. Let $\alpha$ be a generator of $\gf_q^*$ and $\alpha_{i}:=\alpha^{i}$ for $1 \leq i \leq q-1$.
Then $\alpha_{q-1}=1$.
Define
\begin{eqnarray*}\label{eqn-construction4to7}
G_{(i,j,k)}=\left[
\begin{array}{lllll}
1 & 1 & \cdots & 1  & 1\\
\alpha_1^i & \alpha_2^i & \cdots & \alpha_{q-2}^i & \alpha_{q-1}^i\\
\alpha_1^j & \alpha_2^j & \cdots & \alpha_{q-2}^j & \alpha_{q-1}^j\\
\alpha_1^k & \alpha_2^k & \cdots & \alpha_{q-2}^k & \alpha_{q-1}^k\\
\alpha_1^5 & \alpha_2^5 & \cdots & \alpha_{q-2}^5 & \alpha_{q-1}^5
\end{array}
\right],
\end{eqnarray*}
where $(i,j,k)=(2,3,4),(1,2,3),(1,2,4)$ or $(1,3,4)$.  Let $\cC_{(i,j,k)}$ be the linear code over $\gf_q$ generated by $G_{(i,j,k)}$.
We will show that $\cC_{(i,j,k)}$ is an NMDS code and the minimum weight codewords of $\cC_{(i,j,k)}$ and its dual $\cC_{(i,j,k)}^\perp$ support 2-designs. Besides, we denote by $\overline{\cC_{(1,2,4)}}$ the extended code of $\cC_{(1,2,4)}$. We will also prove that the minimum weight codewords of $\overline{\cC_{(1,2,4)}}$ and its dual $\overline{\cC_{(1,2,4)}}^\perp$ support 3-designs.

\subsection{When $(i,j,k)=(2,3,4)$}
Let $(i,j,k)=(2,3,4)$. We study the linear code $\cC_{(2,3,4)}$ in this subsection.

The following lemma plays an important role in the proof of our result.
\begin{lemma}\label{thm-prooft3}
Let $m$ be a positive integer with $m>3$ and $q=2^m$. Let $x_1, x_2, x_3, x_4, x_5$ be five pairwise distinct elements in $\gf_q^*$. Define the matrix
\begin{eqnarray*}\label{eqn-m3}
M_{(2,3,4)}=
\left[
\begin{array}{lllll}
1 & 1 & 1  & 1 & 1\\
x_1^2 & x_2^2 & x_3^2 & x_4^2 & x_5^2\\
x_1^3 & x_2^3 & x_3^3 & x_4^3 & x_5^3\\
x_1^4 & x_2^4 & x_3^4 & x_4^4 & x_5^4\\
x_1^5 & x_2^5 & x_3^5 & x_4^5 & x_5^5
\end{array}\right].
\end{eqnarray*}
Then for any two different and fixed elements $x_1, x_2$, the total number of different choices of $(x_3, x_4, x_5)$ such that $|M_{(2,3,4)}|=0$ is equal to $\frac{(q-4)(q-8)}{6}$ (regardless of the ordering of $x_3, x_4, x_5$).
\end{lemma}

\begin{proof}
By Lemma \ref{lem-Vandermonde}, $|M_{(2,3,4)}|=0$ if and only if
$$x_1x_2x_3x_4+\left(x_1x_2x_3+(x_1x_2+x_1x_3+x_2x_3)x_4\right)x_5=0.$$
Let $x_1,x_2$ be two different and fixed elements in $\gf_q^*$. Consider the following cases.

\emph{Case 1:} Let $x_3=\frac{x_1x_2}{x_1+x_2}$, then $x_1x_2+x_1x_3+x_2x_3=0$ and $x_1x_2x_3+(x_1x_2+x_1x_3+x_2x_3)x_4=x_1x_2x_3 \neq 0.$ Then
$$|M_{(2,3,4)}|=\prod_{1 \leqslant i<j \leqslant 5}(x_j-x_i)(x_1x_2x_3(x_4+x_5)) \neq 0.$$
So there is no  $(x_3, x_4, x_5)$ such that $|M_{(2,3,4)}|=0$ in this case.

\emph{Case 2:} Let $x_3 \neq \frac{x_1x_2}{x_1+x_2}$ and $x_1x_2x_3+(x_1x_2+x_1x_3+x_2x_3)x_4=0$ which implies  $x_4 = \frac{x_1x_2x_3}{x_1x_2+x_1x_3+x_2x_3}$. Then
$$|M_{(2,3,4)}|=\prod_{1 \leqslant i<j \leqslant 5}(x_j-x_i)(x_1x_2x_3x_4) \neq 0.$$
So there is no  $(x_3, x_4, x_5)$ such that $|M_{(2,3,4)}|=0$ in this case.

\emph{Case 3:} Let $x_3 \neq \frac{x_1x_2}{x_1+x_2}$ and $x_1x_2x_3+(x_1x_2+x_1x_3+x_2x_3)x_4 \neq 0 $ which implies $ x_4 \neq \frac{x_1x_2x_3}{x_1x_2+x_1x_3+x_2x_3}$. Then
$$|M_{(2,3,4)}|=0 \Leftrightarrow x_5=\frac{x_1x_2x_3x_4}{x_1x_2x_3+x_1x_2x_4+x_1x_3x_4+x_2x_3x_4}.$$
Since $x_1, x_2, x_3, x_4, x_5$ are five pairwise distinct elements in $\gf_q^*$, then $x_5 \notin \{0,x_1,x_2,x_3,x_4\}$. Similarly to the proof in Lemma \ref{thm-prooft1}, we can derive that  $x_4 \notin \{\frac{x_2x_3}{x_2+x_3},\frac{x_1x_3}{x_1+x_3},\frac{x_1x_2}{x_1+x_2} \}$. We then conclude that $|M_{(2,3,4)}|=0$ if and only
$$x_3 \notin \left\{0, x_1, x_2, \frac{x_1x_2}{x_1+x_2} \right\},$$
$$x_4 \notin \left\{0, x_1, x_2, x_3, \frac{x_2x_3}{x_2+x_3},\frac{x_1x_3}{x_1+x_3},\frac{x_1x_2}{x_1+x_2}, \frac{x_1x_2x_3}{x_1x_2+x_1x_3+x_2x_3} \right\}$$
and
$$x_5=\frac{x_1x_2x_3x_4}{x_1x_2x_3+x_1x_2x_4+x_1x_3x_4+x_2x_3x_4}.$$
It is easy to prove that the elements in
$$ \left\{0, x_1, x_2, \frac{x_1x_2}{x_1+x_2} \right\} $$
are pairwise distinct, so is
$$ \left\{0, x_1, x_2, x_3, \frac{x_2x_3}{x_2+x_3},\frac{x_1x_3}{x_1+x_3},\frac{x_1x_2}{x_1+x_2}, \frac{x_1x_2x_3}{x_1x_2+x_1x_3+x_2x_3} \right\}.$$
Then the total number of different choices of $(x_3, x_4,x_5)$ such that $|M_{(2,3,4)}|=0$ is equal to $\frac{(q-4)(q-8)}{6}$ regardless of the ordering of $x_3, x_4, x_5$.

The desired conclusion follows.
\end{proof}

\begin{theorem}\label{thm-3}
Let $m$ be a positive integer with $m>3$ and $q=2^m$. Then $\cC_{(2,3,4)}$ is a $[q-1, 5, q-6]$ NMDS code over $\gf_q$ with weight enumerator
$$A(z)=1 + \frac{(q-1)^2(q-2)(q-4)(q-8)}{120} z^{q-6} + \frac{5(q-1)^2(q-2)(q-4)}{24} z^{q-5} + $$
$$\frac{(q-1)^2(q-2)(q^2-2q+2)}{12} z^{q-4}+\frac{(q-1)^2(q-2)(2q^2+9q+28)}{12} z^{q-3}+$$
$$\frac{(q-1)^2(9q^3+22q^2+12q+176)}{24} z^{q-2}+\frac{(q-1)(44q^4+65q^3+125q^2-170q+536)}{120}z^{q-1}.$$
Moreover, the minimum weight codewords of $\cC_{(2,3,4)}$ support a 2-$(q-1,q-6,\frac{(q-4)(q-6)(q-7)(q-8)}{120})$ simple design and the minimum weight codewords of $\cC_{(2,3,4)}^\perp$ support a 2-$(q-1,5,\frac{(q-4)(q-8)}{6})$ simple design.
\end{theorem}

\begin{proof}
Similarly to the proof of Theorem \ref{thm-1}, we can derive this Theorem by Lemmas \ref{lem-nmdsweight},\ref{lem-minweight},\ref{thm-prooft3} and Equation (\ref{eqn-t}). The details are omitted here.
\end{proof}

\subsection{When $(i,j,k)=(1,2,3)$}
Let $(i,j,k)=(1,2,3)$. We study the linear code $\cC_{(1,2,3)}$ in this subsection.

The following lemma plays an important role in the proof of our next main result.
\begin{lemma}\label{thm-prooft4}
Let $m$ be a positive integer with $m>3$ and $q=2^m$. Let $x_1, x_2, x_3, x_4, x_5$ be five pairwise distinct elements in $\gf_q^*$. Define a matrix
\begin{eqnarray*}\label{eqn-m4}
M_{(1,2,3)}=
\left[
\begin{array}{lllll}
1 & 1 & 1 & 1 & 1\\
x_1 & x_2 & x_3 & x_4 & x_5\\
x_1^2 & x_2^2 & x_3^2 & x_4^2 & x_5^2\\
x_1^3 & x_2^3 & x_3^3 & x_4^3 & x_5^3\\
x_1^5 & x_2^5 & x_3^5 & x_4^5 & x_5^5
\end{array}\right].
\end{eqnarray*}
Then for any two different and fixed elements $x_1, x_2$, the total number of different choices of $(x_3, x_4, x_5)$ such that $|M_{(1,2,3)}|=0$ is equal to $\frac{(q-4)(q-8)}{6}$ (regardless of the ordering of $x_3, x_4, x_5$).
\end{lemma}

\begin{proof}
Let  $x_1, x_2$ be two  different and fixed elements in $\gf_q^*$.
By Lemma \ref{lem-Vandermonde}, $|M_{(1,2,3)}|=0$ if and only if $x_1+x_2+x_3+x_4+x_5=0$. Then

$$|M_{(1,2,3)}|=0 \Leftrightarrow x_5=x_1+x_2+x_3+x_4. $$
Since  $x_5 \notin \{0,x_1,x_2,x_3,x_4\}$, we  deduce that $x_4 \notin \{x_2+x_3,x_1+x_3,x_1+x_2,x_1+x_2+x_3\}$ and $x_3 \neq x_1+x_2 $. We then conclude that  $|M_{(1,2,3)}|=0$ if and only if
$$x_3 \notin \{0, x_1, x_2, x_1+x_2 \},$$
$$x_4 \notin \{0, x_1, x_2, x_3, x_2+x_3,x_1+x_3,x_1+x_2, x_1+x_2+x_3 \}$$
and
$$x_5=x_1+x_2+x_3+x_4.$$
Obviously, the elements in $ \{0, x_1, x_2, x_1+x_2 \} \mbox{ and } \{0, x_1, x_2, x_3, x_2+x_3,x_1+x_3,x_1+x_2, x_1+x_2+x_3 \}$
are pairwise distinct, respectively. So the total number of different choices of $(x_3, x_4,x_5)$ such that $|M_{(1,2,3)}|=0$ is equal to $\frac{(q-4)(q-8)}{6}$ regardless of the ordering of $x_3, x_4, x_5$.
Then we have completed the proof.
\end{proof}

\begin{theorem}\label{thm-4}
Let $m$ be a positive integer with $m>3$ and $q=2^m$. Then $\cC_{(1,2,3)}$ is a $[q-1, 5, q-6]$ NMDS code over $\gf_q$ with weight enumerator
$$A(z)=1 + \frac{(q-1)^2(q-2)(q-4)(q-8)}{120} z^{q-6} + \frac{5(q-1)^2(q-2)(q-4)}{24} z^{q-5} + $$
$$\frac{(q-1)^2(q-2)(q^2-2q+2)}{12} z^{q-4}+\frac{(q-1)^2(q-2)(2q^2+9q+28)}{12} z^{q-3}+$$
$$\frac{(q-1)^2(9q^3+22q^2+12q+176)}{24} z^{q-2}+\frac{(q-1)(44q^4+65q^3+125q^2-170q+536)}{120}z^{q-1}.$$
Moreover, the minimum weight codewords of $\cC_{(1,2,3)}$ support a 2-$(q-1,q-6,\frac{(q-4)(q-6)(q-7)(q-8)}{120})$ simple design and the minimum weight codewords of $\cC_{(1,2,3)}^\perp$ support a 2-$(q-1,5,\frac{(q-4)(q-8)}{6})$ simple design.
\end{theorem}

\begin{proof}
With a similar proof as that of  Theorem \ref{thm-1}, we can easily prove this theorem by Lemma \ref{thm-prooft4}.
\end{proof}

\subsection{When $(i,j,k)=(1,2,4)$}
In this subsection, let $(i,j,k)=(1,2,4)$ and we study the linear code $\cC_{(1,2,4)}$.

The following lemma is essential for the proof of our result.
\begin{lemma}\label{thm-prooft5}
Let $m$ be an odd integer with $m>3$ and $q=2^m$. Let $x_1, x_2, x_3, x_4, x_5$ be five pairwise distinct elements in $\gf_q^*$. Define the matrix
\begin{eqnarray*}\label{eqn-m5}
M_{(1,2,4)}=
\left[
\begin{array}{lllll}
1 & 1 & 1 & 1 & 1\\
x_1 & x_2 & x_3 & x_4 & x_5\\
x_1^2 & x_2^2 & x_3^2 & x_4^2 & x_5^2\\
x_1^4 & x_2^4 & x_3^4 & x_4^4 & x_5^4\\
x_1^5 & x_2^5 & x_3^5 & x_4^5 & x_5^5
\end{array}\right].
\end{eqnarray*}
Then for any two different and fixed elements $x_1, x_2$, the total number of different choices of $(x_3, x_4, x_5)$ such that $|M_{(1,2,4)}|=0$ is equal to $\frac{(q-5)(q-8)}{6}$ (regardless of the ordering of $x_3,x_4,x_5$).
\end{lemma}

\begin{proof}
By Lemma \ref{lem-Vandermonde}, $|M_{(1,2,4)}|=0$ if and only if
$$x_1x_2+x_1x_3+x_1x_4+x_2x_3+x_2x_4+x_3x_4+(x_1+x_2+x_3+x_4)x_5=0.$$
Let $x_1,x_2$ be two different and fixed elements in $\gf_q^*$. Consider the following cases.

\emph{Case 1:} Let $x_3 = x_1+x_2$. Then
$$|M_{(1,2,4)}|=\prod_{1 \leqslant i<j \leqslant 5}(x_j-x_i)(x_1^2+x_2^2+x_1x_2+x_4x_5) = 0 \Leftrightarrow x_5=\frac{x_1^2+x_2^2+x_1x_2}{x_4}.$$
 Then $x_5 \notin \{0,x_1,x_2,x_3,x_4\}$ implies $x_4 \notin \{x_1+x_2+\frac{x_2^2}{x_1},x_1+x_2+\frac{x_1^2}{x_2},\frac{x_1^2+x_1x_2+x_2^2}{x_1+x_2},a\}$, where $a^2=x_1^2+x_1x_2+x_2^2.$ In this case , we conclude that $|M_{(1,2,4)}|=0$ if and only if
$$ x_4 \notin \left\{0,x_1,x_2,x_1+x_2,x_1+x_2+\frac{x_2^2}{x_1},x_1+x_2+\frac{x_1^2}{x_2},\frac{x_1^2+x_1x_2+x_2^2}{x_1+x_2},a\right\},$$
where $a^2=x_1^2+x_1x_2+x_2^2,$ and
$$ x_5=\frac{x_1^2+x_1x_2+x_2^2}{x_4}. $$
By Lemmas \ref{lem-s} and \ref{lem-ab}, we can easily prove that the elements in
$$\left\{0,x_1,x_2,x_1+x_2,x_1+x_2+\frac{x_2^2}{x_1},x_1+x_2+\frac{x_1^2}{x_2},\frac{x_1^2+x_1x_2+x_2^2}{x_1+x_2},a\right\}$$
are pairwise distinct.
In this case,  the total number of different choices of $(x_3, x_4,x_5)$ such that $|M_{(1,2,4)}|=0$ is equal to $\frac{(q-8)}{6}$ regardless of the ordering of $x_3,x_4,x_5$.

\emph{Case 2:} Let $x_3 \neq x_1+x_2$ and $x_1+x_2+x_3+x_4=0$.  Then $x_4=x_1+x_2+x_3$. By Lemma \ref{lem-ab},
$$|M_{(1,2,4)}|=\prod_{1 \leqslant i<j \leqslant 5}(x_j-x_i)(x_1^2+x_2^2+x_3^2+x_1x_2+x_1x_3+x_2x_3) \neq 0.$$
So there is no  $(x_3, x_4, x_5)$ such that $|M_{(1,2,4)}|=0$ in this case.

\emph{Case 3:} Let $x_3 \neq x_1+x_2$ and $x_1+x_2+x_3+x_4\neq0$. Then $x_4 \neq x_1+x_2+x_3$ and
$$|M_{(1,2,4)}|=0 \Leftrightarrow  x_5=\frac{x_1x_2+x_1x_3+x_1x_4+x_2x_3+x_2x_4+x_3x_4}{x_1+x_2+x_3+x_4}.$$
It is easy to deduce that $x_5 \notin \{0,x_1,x_2,x_3,x_4\}$ implies
$$x_4 \notin \left\{\frac{x_1x_2+x_1x_3+x_2x_3}{x_1+x_2+x_3},\frac{x_1^2+x_2x_3}{x_2+x_3},\frac{x_2^2+x_1x_3}{x_1+x_3},\frac{x_3^2+x_1x_2}{x_1+x_2}, b\right\},$$
where $b^2=x_1x_2+x_1x_3+x_2x_3$. In this case, we conclude that  $|M_{(1,2,4)}|=0$ if and only if
$$x_3 \notin \{0, x_1, x_2, x_1+x_2 \},$$
$$x_4 \notin L:= \left\{0, x_1, x_2, x_3, x_1+x_2+x_3 , \frac{x_1x_2+x_1x_3+x_2x_3}{x_1+x_2+x_3},\frac{x_1^2+x_2x_3}{x_2+x_3},\frac{x_2^2+x_1x_3}{x_1+x_3},\frac{x_3^2+x_1x_2}{x_1+x_2},b \right\},$$
where $b^2=x_1x_2+x_1x_3+x_2x_3$, and
$$x_5=\frac{x_1x_2+x_1x_3+x_1x_4+x_2x_3+x_2x_4+x_3x_4}{x_1+x_2+x_3+x_4}.$$
Consider the following subcases of $L$.

\emph{Subcase 3.1:} If $x_3=\frac{x_1^2}{x_2}$, then $\frac{x_1^2+x_2x_3}{x_2+x_3}=0$, $\frac{x_1x_2+x_1x_3+x_2x_3}{x_1+x_2+x_3}=x_1$ and other elements in $L$ are pairwise distinct, which implies $|L|=8$. If $x_3=\frac{x_2^2}{x_1}$, then by the symmetry of $x_1$ and $x_2$, we have $|L|=8$.

\emph{Subcase 3.2:} If $x_3=c$ for $c^2=x_1x_2$, then $\frac{x_3^2+x_1x_2}{x_1+x_2}=0$, $\frac{x_1x_2+x_1x_3+x_2x_3}{x_1+x_2+x_3}=x_3$ and the other elements in $L$ are pairwise distinct, which implies $|L|=8$.

\emph{Subcase 3.3:} If $x_3=\frac{x_1x_2}{x_1+x_2}$, then $\frac{x_1x_2+x_1x_3+x_2x_3}{x_1+x_2+x_3}=b=0$ and other elements in $L$ are pairwise distinct, which implies $|L|=8$.

\emph{Subcase 3.4:} Let $x_3 \notin S:=\left\{0, x_1, x_2, x_1+x_2,\frac{x_1^2}{x_2},\frac{x_2^2}{x_1},\frac{x_1x_2}{x_1+x_2},c \right\}$, where $c^2=x_1x_2$. By Lemmas \ref{lem-s} and \ref{lem-ab}, it is easy to prove that the elements in $S$ are pairwise distinct and $|S|=8.$ By Lemmas \ref{lem-s} and \ref{lem-abc}, the elements in $L$ are pairwise distinct, which implies $|L|=10$.

In this case, the total number of different choices of $(x_3, x_4,x_5)$ such that $|M_{(1,2,4)}|=0$ is equal to
$$\frac{4(q-8)}{3!}+\frac{(q-8)(q-10)}{3!}=\frac{(q-6)(q-8)}{6}$$
regardless of the ordering of $x_3,x_4,x_5$.

Thanks to the above cases,  the total number of different choices of $(x_3, x_4,x_5)$ such that $|M_{(1,2,4)}|=0$ is equal to
$$\frac{(q-8)}{6}+\frac{(q-6)(q-8)}{6}=\frac{(q-5)(q-8)}{6}$$
regardless of the ordering of $x_3,x_4,x_5$.

Then we have completed the proof.
\end{proof}

\begin{theorem}\label{thm-5}
Let $m$ be an odd integer with $m>3$ and $q=2^m$. Then $\cC_{(1,2,4)}$ is a $[q-1, 5, q-6]$ NMDS code over $\gf_q$ with weight enumerator
$$A(z)=1 + \frac{(q-1)^2(q-2)(q-5)(q-8)}{120} z^{q-6} + \frac{(q-1)^2(q-2)(3q-14)}{12} z^{q-5} + $$
$$\frac{(q-1)^2(q-2)(q^2-3q+10)}{12} z^{q-4}+\frac{(q-1)^2(q-2)(q^2+5q+10)}{6} z^{q-3}+$$
$$\frac{(q-1)^2(9q^3+21q^2+22q+160)}{24} z^{q-2}+\frac{(q-1)(22q^4+33q^3+57q^2-72q+260)}{60}z^{q-1}.$$
Moreover, the minimum weight codewords of $\cC_{(1,2,4)}$ support a 2-$(q-1,q-6,\frac{(q-5)(q-6)(q-7)(q-8)}{120})$ simple design and the minimum weight codewords of $\cC_{(1,2,4)}^\perp$ support a 2-$(q-1,5,\frac{(q-5)(q-8)}{6})$ simple design.
\end{theorem}

\begin{proof}
By Lemma \ref{thm-prooft5}, we can prove this theorem with a similar proof as that of Theorem \ref{thm-1}.
The details are omitted here.
\end{proof}

\subsection{The extended code $\overline{\cC_{(1,2,4)}}$}
It is obvious that the extended  code $\overline{\cC_{(1,2,4)}}$ of $\cC_{(1,2,4)}$ is generated by the following matrix:

\begin{eqnarray*}\label{eqn-construction-extended}
\overline{G_{(1,2,4)}}=\left[
\begin{array}{lllll}
1 & 1 & \cdots & 1  & 1\\
\alpha_1^1 & \alpha_2^1 & \cdots & \alpha_{q-1}^1 & 0\\
\alpha_1^2 & \alpha_2^2 & \cdots & \alpha_{q-1}^2 & 0\\
\alpha_1^4 & \alpha_2^4 & \cdots & \alpha_{q-1}^4 & 0\\
\alpha_1^5 & \alpha_2^5 & \cdots & \alpha_{q-1}^5 & 0
\end{array}
\right].
\end{eqnarray*}

We need the following lemma to give our main result in this subsection.
\begin{lemma}\label{thm-proof-extended}
Let $m$ be an odd integer with $m>3$ and $q=2^m$. Let $x_1, x_2, x_3, x_4, x_5$ be five pairwise distinct elements in $\gf_q$. Define the matrix
\begin{eqnarray*}\label{eqn-m5}
\overline{M_{(1,2,4)}}=
\left[
\begin{array}{lllll}
1 & 1 & 1 & 1 & 1\\
x_1 & x_2 & x_3 & x_4 & x_5\\
x_1^2 & x_2^2 & x_3^2 & x_4^2 & x_5^2\\
x_1^4 & x_2^4 & x_3^4 & x_4^4 & x_5^4\\
x_1^5 & x_2^5 & x_3^5 & x_4^5 & x_5^5
\end{array}\right].
\end{eqnarray*}
Then for any pairwise different and fixed elements $x_1, x_2, x_3$, the total number of different choices of $(x_4, x_5)$ such that $|\overline{M_{(1,2,4)}}|=0$ is equal to $\frac{q-8}{2}$ (regardless of the ordering of $x_4, x_5$).
\end{lemma}

\begin{proof}
By Lemmas \ref{lem-abc}, \ref{lem-ab} and \ref{lem-Vandermonde}, we can prove this theorem with a similar proof as that of Lemma \ref{thm-prooft1}.
\end{proof}

\begin{theorem}\label{thm-extended}
Let $m$ be an odd integer with $m>3$ and $q=2^m$. Then $\overline{\cC_{(1,2,4)}}$  is a $[q, 5, q-5]$ NMDS code over $\gf_q$ with weight enumerator
$$A(z)=1 + \frac{q(q-1)^2(q-2)(q-8)}{120} z^{q-5} + \frac{5q(q-1)^2(q-2)}{24} z^{q-4} + $$
$$\frac{q^2(q-1)^2(q-2)}{12} z^{q-3}+\frac{q(q-1)^2(2q^2+7q+20)}{12} z^{q-2}+$$
$$\frac{q(q-1)(9q^3+13q^2-6q+80)}{24} z^{q-1}+\frac{(q-1)(44q^4+21q^3+49q^2-114q+120)}{120}z^{q}.$$
Moreover, the minimum weight codewords of $\overline{\cC_{(1,2,4)}}$ support a 3-$(q,q-5,\frac{(q-5)(q-6)(q-7)(q-8)}{120})$ simple design and the minimum weight codewords of $\overline{\cC_{(1,2,4)}}^\perp$ support a 3-$(q,5,\frac{q-8}{2})$ simple design.
\end{theorem}

\begin{proof}
By Lemma \ref{thm-proof-extended}, we can derive this theorem with a similar proof as that of Theorem \ref{thm-1}.
The details are omitted here.
\end{proof}

\subsection{When $(i,j,k)=(1,3,4)$}
Let $(i,j,k)=(1,3,4)$ and we study the linear code $\cC_{(1,3,4)}$ in this subsection.

The following lemma will be used to prove the main result in this subsection.
\begin{lemma}\label{thm-prooft6}
Let $m$ be an odd integer with $m>3$ and $q=2^m$. Let $x_1, x_2, x_3, x_4, x_5$ be five pairwise distinct elements in $\gf_q^*$. Define the matrix
\begin{eqnarray*}\label{eqn-m6}
M_{(1,3,4)}=
\left[
\begin{array}{lllll}
1 & 1 & 1 & 1 & 1\\
x_1 & x_2 & x_3 & x_4 & x_5\\
x_1^3 & x_2^3 & x_3^3 & x_4^3 & x_5^3\\
x_1^4 & x_2^4 & x_3^4 & x_4^4 & x_5^4\\
x_1^5 & x_2^5 & x_3^5 & x_4^5 & x_5^5
\end{array}\right].
\end{eqnarray*}
Then for any two different and fixed elements $x_1, x_2$, the total number of different choices of $(x_3, x_4, x_5)$ such that $|M_{(1,3,4)}|=0$ is equal to $\frac{(q-5)(q-8)}{6}$ (regardless of the ordering of $x_3, x_4, x_5$).
\end{lemma}

\begin{proof}
Let $x_1, x_2$ be two different and fixed elements.
By Lemma \ref{lem-Vandermonde}, $|M_{(1,3,4)}|=0$ if and only if
 $$x_1x_2x_3+x_1x_2x_4+x_1x_3x_4+x_2x_3x_4+(x_1x_2+x_1x_3+x_2x_3+x_1x_4+x_2x_4+x_3x_4)x_5=0.$$
 Consider the following cases.

\emph{Case 1:} Let $x_3 = x_1+x_2$. Then $x_1x_2+x_1x_3+x_2x_3+(x_1+x_2+x_3)x_4=x_1^2+x_1x_2+x_2^2 \neq 0$ and
$$|M_{(1,3,4)}| = 0 \Leftrightarrow x_5=\frac{(x_1+x_2)x_1x_2}{x_1^2+x_1x_2+x_2^2}+x_4.$$
 It is easy to deduce that
 $x_5 \notin \{0,x_1,x_2,x_3,x_4\}$ if and only if
 $$x_4 \notin \left\{ \frac{(x_1+x_2)x_1x_2}{x_1^2+x_1x_2+x_2^2},\\ \frac{x_1^3}{x_1^2+x_1x_2+x_2^2}, \frac{x_2^3}{x_1^2+x_1x_2+x_2^2},\frac{(x_1+x_2)^3}{x_1^2+x_1x_2+x_2^2} \right\}.$$
 In this case, we conclude that  $|M_{(1,3,4)}|=0$ if and only if
$$ x_4 \notin \left\{0,x_1,x_2,x_1+x_2,\frac{(x_1+x_2)x_1x_2}{x_1^2+x_1x_2+x_2^2}, \frac{x_1^3}{x_1^2+x_1x_2+x_2^2}, \frac{x_2^3}{x_1^2+x_1x_2+x_2^2},\frac{(x_1+x_2)^3}{x_1^2+x_1x_2+x_2^2} \right\}$$
and
$$ x_5=\frac{(x_1+x_2)x_1x_2}{x_1^2+x_1x_2+x_2^2}+x_4. $$
By Lemmas \ref{lem-s} and \ref{lem-ab}, we can easily prove that the elements in
$$\left\{0,x_1,x_2,x_1+x_2,\frac{(x_1+x_2)x_1x_2}{x_1^2+x_1x_2+x_2^2}, \frac{x_1^3}{x_1^2+x_1x_2+x_2^2}, \frac{x_2^3}{x_1^2+x_1x_2+x_2^2},\frac{(x_1+x_2)^3}{x_1^2+x_1x_2+x_2^2} \right\}$$
are pairwise distinct.
Then the total number of different choices of $(x_3, x_4,x_5)$ such that $|M_{(1,3,4)}|=0$ is equal to $\frac{q-8}{6}$ regardless of the ordering of $x_3, x_4, x_5$.

\emph{Case 2:} Let $x_3 \neq x_1+x_2$ and $x_1x_2+x_1x_3+x_2x_3+(x_1+x_2+x_3)x_4=0$ which implies $x_4=\frac{x_1x_2+x_1x_3+x_2x_3}{x_1+x_2+x_3}$. By Lemma \ref{lem-abc}, it is easy to prove
$$|M_{(1,3,4)}|=\prod_{1 \leqslant i<j \leqslant 5}(x_j-x_i)\frac{x_1^2(x_2^2+x_2x_3+x_3^2)+x_1(x_2^2x_3+x_2x_3^2)+x_2^2x_3^2}{x_1+x_2+x_3} \neq 0.$$
Hence there is no $(x_3, x_4,x_5)$ such that $|M_{(1,3,4)}|=0$.

\emph{Case 3:} Let $x_3 \neq x_1+x_2$ and $x_1x_2+x_1x_3+x_2x_3+(x_1+x_2+x_3)x_4 \neq 0$ which implies  $x_4 \neq \frac{x_1x_2+x_1x_3+x_2x_3}{x_1+x_2+x_3}.$
Then
\begin{eqnarray}\label{eqn-case}
|M_{(1,3,4)}|= 0 \Leftrightarrow x_5=\frac{(x_1x_2+x_1x_3+x_2x_3)x_4+x_1x_2x_3}{x_1x_2+x_1x_3+x_2x_3+(x_1+x_2+x_3)x_4}.
\end{eqnarray}
Let
$$x_5=\frac{(x_1x_2+x_1x_3+x_2x_3)x_4+x_1x_2x_3}{x_1x_2+x_1x_3+x_2x_3+(x_1+x_2+x_3)x_4},\ a^2=\frac{x_1x_2x_3}{x_1+x_2+x_3}.$$
Consider the following subcases.

\emph{Subcase 3.1:} Let $x_3 = \frac{x_1x_2}{x_1+x_2}$, then $x_5=\frac{x_1^2x_2^2}{(x_1^2+x_1x_2+x_2^2)x_4}$ and $\frac{x_1x_2+x_1x_3+x_2x_3}{x_1+x_2+x_3}=0$.
It is easy to deduce that $x_5 \notin \{0,x_1,x_2,x_3,x_4\}$ implies  $x_4 \notin\left\{\ \frac{x_1x_2^2}{x_1^2+x_1x_2+x_2^2}, \frac{x_1^2x_2}{x_1^2+x_1x_2+x_2^2}, \frac{x_1x_2(x_1+x_2)}{x_1^2+x_1x_2+x_2^2}, a \right\}$. In this subcase, we conclude that $|M_{(1,3,4)}|=0$ if and only if
$$ x_4 \notin L_1 :=\left \{0,x_1,x_2,\frac{x_1x_2}{x_1+x_2}, \frac{x_1x_2^2}{x_1^2+x_1x_2+x_2^2}, \frac{x_1^2x_2}{x_1^2+x_1x_2+x_2^2}, \frac{x_1x_2(x_1+x_2)}{x_1^2+x_1x_2+x_2^2}, a\right\},$$
and
$$ x_5=\frac{x_1^2x_2^2}{(x_1^2+x_1x_2+x_2^2)x_4}. $$
It is obvious that the elements in $L_1$ are pairwise distinct.
Then the total number of different choices of $(x_3, x_4,x_5)$ such that $|M_{(1,3,4)}|=0$ is equal to $\frac{q-8}{6}$ regardless of the ordering of $x_3,x_4,x_5$.

\emph{Subcase 3.2:} Let $x_3 = \frac{x_1^2}{x_2}$, then $x_5=\frac{x_1(x_4(x_1+x_2)^2+x_1x_2(x_1+x_4))}{x_4(x_1+x_2)^2+x_1x_2(x_1+x_4)+x_1(x_1+x_2)^2}$ and $\frac{x_1x_2+x_1x_3+x_2x_3}{x_1+x_2+x_3}=x_1$. It is easy to deduce that
$x_5 \notin \{0,x_1,x_2,x_3,x_4\}$ if and only if $ x_4 \notin \{ \frac{x_1x_2^2}{x_1^2+x_1x_2+x_2^2}, \frac{x_1^3}{x_1^2+x_1x_2+x_2^2}, \\\frac{x_1^2x_2}{x_1^2+x_1x_2+x_2^2},a \}$.
We conclude that $|M_{(1,3,4)}|=0$ if and only if
$$ x_4 \notin L_2: = \left\{0, x_1, x_2, \frac{x_1^2}{x_2}, \frac{x_1x_2^2}{x_1^2+x_1x_2+x_2^2}, \frac{x_1^3}{x_1^2+x_1x_2+x_2^2}, \frac{x_1^2x_2}{x_1^2+x_1x_2+x_2^2},a\right\}$$
and
$$x_5=\frac{x_1(x_4(x_1+x_2)^2+x_1x_2(x_1+x_4))}{x_4(x_1+x_2)^2+x_1x_2(x_1+x_4)+x_1(x_1+x_2)^2}.$$
It is obvious that the elements in $L_2$ are pairwise distinct.
In this subcase, the total number of choices of $x_3, x_4,x_5$ such that $|M_{(1,3,4)}|=0$ is equal to $\frac{q-8}{6}$ regardless of the ordering of $x_3,x_4,x_5$.

Let $x_3 = \frac{x_2^2}{x_1}$, then by the symmetry of $x_1$ and $x_2$, we can drive the same conclusion that the total number of choices of $x_3, x_4,x_5$ such that $|M_{(1,3,4)}|=0$ is equal to $\frac{q-8}{6}$ regardless of the ordering of $x_3,x_4,x_5$.

\emph{Subcase 3.3:} Let $x_3 = b$, where $b^2=x_1x_2$.  Then $x_5^2=\frac{x_1x_2(x_4^2(x_1^2+x_1x_2+x_2^2)+x_1^2x_2^2)}{x_4^2(x_1^2+x_1x_2+x_2^2)+x_1^2x_2^2+x_1x_2(x_1+x_2)^2}$. It is easy to deduce that $x_5 \notin \{0,x_1,x_2,x_3,x_4\}$ implies
$$ x_4^2 \notin \left\{ \frac{x_1^3x_2}{x_1^2+x_1x_2+x_2^2}, \frac{x_2^3x_1}{x_1^2+x_1x_2+x_2^2}, \frac{x_1^2x_2^2}{x_1^2+x_1x_2+x_2^2}, \frac{x_1x_2b}{x_1+x_2+b} \right\}.$$
We conclude that  $|M_{(1,3,4)}|=0$ if and only if
$$ x_4^2 \notin L_4:= \left\{0, x_1^2, x_2^2, x_1x_2, \frac{x_1^3x_2}{x_1^2+x_1x_2+x_2^2}, \frac{x_2^3x_1}{x_1^2+x_1x_2+x_2^2}, \frac{x_1^2x_2^2}{x_1^2+x_1x_2+x_2^2}, \frac{x_1x_2b}{x_1+x_2+b} \right\}$$
and
$$x_5^2=\frac{x_1x_2(x_4^2(x_1^2+x_1x_2+x_2^2)+x_1^2x_2^2)}{x_4^2(x_1^2+x_1x_2+x_2^2)+x_1^2x_2^2+x_1x_2(x_1+x_2)^2}.$$
By Lemma \ref{lem-abc}, we can verify that the elements in $L_4$ are pairwise distinct.
In this subcase, the total number of different choices of $(x_3, x_4,x_5)$ such that $|M_{(1,3,4)}|=0$ is equal to $\frac{q-8}{6}$ regardless of the ordering of $x_3,x_4,x_5$.

\emph{Subcase 3.4:} Let $x_3 \notin L_5:=\left\{0,x_1,x_2,x_1+x_2,\frac{x_1x_2}{x_1+x_2},\frac{x_1^2}{x_2},\frac{x_2^2}{x_1},b\right\}$, where $b^2=x_1x_2$. By Lemmas \ref{lem-s} and \ref{lem-ab}, the elements in $L_5$ are pairwise distinct. It is easy  to prove that
$x_5 \notin \{0,x_1,x_2,x_3,x_4\}$ implies  $ x_4 \notin \left\{ \frac{x_1^2(x_2+x_3)}{x_1^2+x_2x_3}, \frac{x_2^2(x_1+x_3)}{x_2^2+x_1x_3}, \frac{x_1x_2x_3}{x_1+x_2+x_3}, \frac{x_3^2(x_1+x_2)}{x_3^2+x_1x_2}, a \right\}$.
We conclude that  $|M_{(1,3,4)}|=0$ if and only if $x_3 \notin L_5$ and $x_4 \notin L_6$, where $L_6$ is given by
\begin{small}
\begin{eqnarray*}
\left\{0, x_1, x_2, x_3, \frac{x_1x_2+x_1x_3+x_2x_3}{x_1+x_2+x_3},\frac{x_1x_2x_3}{x_1x_2+(x_1+x_2)x_3}, \frac{x_1^2(x_2+x_3)}{x_1^2+x_2x_3},  \frac{x_2^2(x_1+x_3)}{x_2^2+x_1x_3},
\frac{x_3^2(x_1+x_2)}{x_3^2+x_1x_2},\frac{x_1x_2x_3}{x_1+x_2+x_3}\right\},
\end{eqnarray*}
\end{small}
and
$$x_5=\frac{(x_1x_2+x_1x_3+x_2x_3)x_4+x_1x_2x_3}{x_1x_2+x_1x_3+x_2x_3+(x_1+x_2+x_3)x_4}.$$
By Lemma \ref{lem-abc}, we can verify that the elements in $L_6$ are pairwise distinct.
In this subcase, the total number of different choices of $(x_3, x_4,x_5)$ such that $|M_{(1,3,4)}|=0$ is equal to $\frac{(q-8)(q-10)}{6}$ regardless of the ordering of $x_3,x_4,x_5$.

In Case 3, the total number of different choices of $(x_3, x_4,x_5)$ such that $|M_{(1,3,4)}|=0$ is equal to
$$\frac{4(q-8)}{6}+\frac{(q-8)(q-10)}{6}=\frac{(q-6)(q-8)}{6}$$
regardless of the ordering of $x_3,x_4,x_5$.

Thanks to the above cases,  the total number of different choices of $(x_3, x_4,x_5)$ such that $|M_{(1,3,4)}|=0$ is equal to
$$\frac{(q-8)}{6}+\frac{(q-6)(q-8)}{6}=\frac{(q-5)(q-8)}{6}$$
regardless of the ordering of $x_3,x_4,x_5$. Then we complete the proof.
\end{proof}

\begin{theorem}\label{thm-6}
Let $m$ be an odd integer with $m>3$ and $q=2^m$. Then $\cC_{(1,3,4)}$ is a $[q-1, 5, q-6]$ NMDS code over $\gf_q$ with weight enumerator
$$A(z)=1 + \frac{(q-1)^2(q-2)(q-5)(q-8)}{120} z^{q-6} + \frac{(q-1)^2(q-2)(3q-14)}{12} z^{q-5} + $$
$$\frac{(q-1)^2(q-2)(q^2-3q+10)}{12} z^{q-4}+\frac{(q-1)^2(q-2)(q^2+5q+10)}{6} z^{q-3}+$$
$$\frac{(q-1)^2(9q^3+21q^2+22q+160)}{24} z^{q-2}+\frac{(q-1)(22q^4+33q^3+57q^2-72q+260)}{60}z^{q-1}.$$
Moreover, the minimum weight codewords of $\cC_{(1,3,4)}$ support a 2-$(q-1,q-6,\frac{(q-5)(q-6)(q-7)(q-8)}{120})$ simple design and the minimum weight codewords of $\cC_{(1,3,4)}^\perp$ support a 2-$(q-1,5,\frac{(q-5)(q-8)}{6})$ simple design.
\end{theorem}

\begin{proof}
By Lemma \ref{thm-prooft6}, the desired conclusions can be derived with a similar proof as that of Theorem \ref{thm-1}.
\end{proof}

\section{Constructions of 6-dimensional near MDS codes holding 2-designs or 3-designs}\label{sect-code4}

In this section, let $q=2^m$ with $m > 3$. Let $\alpha$ be a generator of $\gf_q^*$ and $\alpha_{t}=\alpha^{t}$ for $1 \leq t \leq q-1$.
Then $\alpha_{q-1}=1$.
Define
\begin{eqnarray*}\label{eqn-construction8}
H_{(i,j,k,l)}=\left[
\begin{array}{lllll}
1 & 1 & \cdots & 1  & 1\\
\alpha_1^i & \alpha_2^i & \cdots & \alpha_{q-2}^i & \alpha_{q-1}^i\\
\alpha_1^j & \alpha_2^j & \cdots & \alpha_{q-2}^j & \alpha_{q-1}^j\\
\alpha_1^k & \alpha_2^k & \cdots & \alpha_{q-2}^k & \alpha_{q-1}^k\\
\alpha_1^l & \alpha_2^l & \cdots & \alpha_{q-2}^l & \alpha_{q-1}^l\\
\alpha_1^6 & \alpha_2^6 & \cdots & \alpha_{q-2}^6 & \alpha_{q-1}^6
\end{array}
\right],
\end{eqnarray*}
where $(i,j,k,l)\in \{(1,2,3,4),(2,3,4,5),(1,2,4,5)\}$.
Then $H_{(i,j,k,l)}$ is a $6$ by $q-1$ matrix over $\gf_q$. Let $\cC_{(i,j,k,l)}$ be the linear code over $\gf_q$ generated by $H_{(i,j,k,l)}$.
We will prove that $\cC_{(i,j,k,l)}$ is an NMDS code and the minimum weight codewords of both $\cC_{(i,j,k,l)}$ and its dual $\cC_{(i,j,k,l)}^\perp$ support 2-designs. Let $\overline{\cC_{(i,j,k,l)}}$ be the extended code of $\cC_{(i,j,k,l)}$. We will also prove that $\overline{\cC_{(i,j,k,l)}}$ is an NMDS code and the minimum weight codewords of both $\overline{\cC_{(i,j,k,l)}}$ and its dual $\overline{\cC_{(i,j,k,l)}}^\perp$ support 3-designs for $(i,j,k,l)\in \{(1,2,3,4),(2,3,4,5)\}$.

\subsection{When $(i,j,k,l)\in \{(1,2,3,4),(2,3,4,5)$\}}

The following lemma plays an important role in the proof of our main result in this subsection.
\begin{lemma}\label{thm-prooft7}
Let $m$ be an integer with $m>3$ and $q=2^m$. Let $x_1, x_2, x_3, x_4, x_5, x_6$ be six pairwise distinct elements in $\gf_q^*$. Define the matrix
\begin{eqnarray*}\label{eqn-m7}
M_{(i,j,k,l)}=
\left[
\begin{array}{llllll}
1 & 1 & 1 & 1 & 1 & 1\\
x_1^i & x_2^i & x_3^i & x_4^i & x_5^i & x_6^i\\
x_1^j & x_2^j & x_3^j & x_4^j & x_5^j & x_6^j\\
x_1^k & x_2^k & x_3^k & x_4^k & x_5^k & x_6^k\\
x_1^l & x_2^l & x_3^l & x_4^l & x_5^l & x_6^l\\
x_1^6 & x_2^6 & x_3^6 & x_4^6 & x_5^6 & x_6^6
\end{array}\right],
\end{eqnarray*}
where $(i,j,k,l)=(1,2,3,4)$ or $(2,3,4,5)$.
Then for any two different and fixed elements $x_1, x_2$ and any $(i,j,k,l)\in \{(1,2,3,4),(2,3,4,5)\}$, the total number of different choices of $(x_3, x_4, x_5,x_6)$ such that $|M_{(i,j,k,l)}|=0$ is equal to $\frac{(q-4)(q-6)(q-8)}{24}$ (regardless of the ordering of $x_3, x_4, x_5,x_6$).
\end{lemma}

\begin{proof}
In the following, we only give the proof for $(i,j,k,l)=(1,2,3,4)$ as the proof for $(i,j,k,l)=(2,3,4,5)$ can be similarly given.

Let $x_1, x_2$ be any two different and fixed elements in $\gf_q^*$.
By Lemma \ref{lem-Vandermonde}, $|M_{(1,2,3,4)}|=0$ if and only if $x_6=x_1+x_2+x_3+x_4+x_5.$ It is easy to deduce that $x_6 \notin \{0,x_1,x_2,x_3,x_4,x_5 \}$ implies
$x_5 \notin \{x_1+x_2+x_3+x_4, x_2+x_3+x_4, x_1+x_3+x_4, x_1+x_2+x_4, x_1+x_2+x_3 \}$ and $ x_4 \neq x_1+x_2+x_3 $. In conclusion, $|M_{(1,2,3,4)}|=0$ if and only if $x_3 \notin \{0, x_1, x_2 \}$, $x_4 \notin L_1:= \{0,x_1,x_2,x_3, x_1+x_2+x_3\}$
and
$$x_5 \notin L_2:=\{0,x_1,x_2,x_3,x_4,x_1+x_2+x_3+x_4, x_2+x_3+x_4, x_1+x_3+x_4, x_1+x_2+x_4, x_1+x_2+x_3 \}.$$
Consider the following cases.

\emph{Case 1:} Let $x_3=x_1+x_2$. Then $x_1+x_2+x_3=0$ and $x_1+x_2+x_3+x_4=x_4$. We have $|L_1|=4,|L_2|=8$. In this case, the total number of different choices of $(x_3, x_4,x_5,x_6)$ such that $|M_{(1,2,3,4)}|=0$ is equal to $\frac{(q-4)(q-8)}{24}$ regardless of the ordering of $x_3, x_4, x_5,x_6$.

\emph{Case 2:} Let $x_3 \neq x_1+x_2$. It is obvious that the elements in $L_1$ are pairwise distinct.
Consider the following subcases of $L_2$.

\emph{Subcase 2.1:} Let $x_4 = x_1+x_2$. Then $x_1+x_2+x_3+x_4=x_3$, $ x_1+x_2+x_4=0$ and elements in $L_2$ are pairwise distinct. Thus $|L_2|=8.$ If $x_4 = x_1+x_3$ or $x_4 = x_2+x_3$, then by the symmetry of $x_1, x_2,$ and $x_3$, we also have $|L_2|=8.$

\emph{Subcase 2.2:} Let $x_4 \notin L_3=\{0,x_1,x_2,x_3, x_1+x_2+x_3, x_1+x_2, x_1+x_3,x_2+x_3\}$. Then the elements in $L_2$ and $L_3$ are pairwise distinct. Thus $|L_2|=10.$

By summarizing the four subcases in Case 2, the total number of different choices of $(x_3, x_4,x_5,x_6)$ such that $|M_{(1,2,3,4)}|=0$ is equal to
$$ \frac{3(q-4)(q-8)}{4!} + \frac{(q-4)(q-8)(q-10)}{4!} = \frac{(q-4)(q-7)(q-8)}{24}$$
regardless of the ordering of $x_3, x_4, x_5,x_6$.

By the above cases, the total number of different choices of $(x_3, x_4, x_5,x_6)$ such that $|M_{(1,2,3,4)}|=0$ is equal to
$$\frac{(q-4)(q-8)}{24} + \frac{(q-4)(q-7)(q-8)}{24} = \frac{(q-4)(q-6)(q-8)}{24}$$
regardless of the ordering of $x_3, x_4, x_5,x_6$.
Then the proof is completed.
\end{proof}

\begin{theorem}\label{thm-7}
Let $m$ be a integer with $m>3$ and $q=2^m$. Let $(i,j,k,l)=(1,2,3,4)$ or $(2,3,4,5)$. Then $\cC_{(i,j,k,l)}$ is a $[q-1, 6, q-7]$ NMDS code over $\gf_q$ with weight enumerator
$$A(z)=1 + \frac{(q-1)^2(q-2)(q-4)(q-6)(q-8)}{720} z^{q-7} + \frac{(q-1)^2(q-2)(q-4)(2q-11)}{40} z^{q-6} + $$
$$\frac{(q-1)^2(q-2)(q-4)(q^2-2q+12)}{48} z^{q-5}+\frac{(q-1)^2(q-2)(2q^3+6q^2-5q-78)}{36} z^{q-4}+$$
$$\frac{(q-1)^2(q-2)(q+4)(3q^2-2q+24)}{16} z^{q-3}+\frac{(q-1)^2(44q^4+110q^3+235q^2-110q+1416)}{120}z^{q-2}$$
$$+\frac{(q-1)(265q^5+399q^4+400q^3+1200q^2-1880q+3936)}{720}z^{q-1}.  $$
Moreover, the minimum weight codewords of $\cC_{(i,j,k,l)}$ support a 2-$(q-1,q-7,\frac{(q-4)(q-6)(q-7)(q-8)^2}{720})$ simple design and the minimum weight codewords of $\cC_{(i,j,k,l)}^\perp$ support a 2-$(q-1,6,\frac{(q-4)(q-6)(q-8)}{24})$ simple design.
\end{theorem}

\begin{proof}
By Lemma  \ref{thm-prooft7}, we can prove this theorem with a similar proof as that of Theorem \ref{thm-1}.
\end{proof}

\begin{remark}
Note that $\cC_{(1,2,3)}$ and $\cC_{(2,3,4)}$, $\cC_{(1,2,4)}$ and $\cC_{(1,3,4)}$, $\cC_{(1,2,3,4)}$ and $\cC_{(2,3,4,5)}$ have the same parameters and weight enumerators, respectively. Besides, the minimum weight codewords of each pair of NMDS codes hold $t$-designs with the same parameters, respectively.
However, by Magma, the blocks of the $t$-designs are different for each pair of NMDS codes.
It is open whether  the NMDS codes in each pair are equivalent to each other.
\end{remark}

By Theorems \ref{thm-4} and \ref{thm-7}, we have the following conjecture.
\begin{conj}\label{conj-1}
Let $q=2^m$ with $m \geq k_1\geq 3$, where $k_1$ is some proper positive integer. Let $\alpha_1,\alpha_2, \cdots , \alpha_{q-1}$ be all elements of $\gf_q^*$.
Define
\begin{eqnarray*}\label{eqn-conj1}
M_k
=\left[
\begin{array}{ccccc}
1 & 1 & \cdots & 1  & 1\\
\alpha_1 & \alpha_2 & \cdots & \alpha_{q-2} & \alpha_{q-1}\\
\alpha_1^2 & \alpha_2^2 & \cdots & \alpha_{q-2}^2 & \alpha_{q-1}^2\\
 \vdots   & \vdots   & \vdots                &  \vdots     & \vdots          \\
\alpha_1^{k-2} & \alpha_2^{k-2} & \cdots & \alpha_{q-2}^{k-2} & \alpha_{q-1}^{k-2}\\
\alpha_1^k & \alpha_2^k & \cdots & \alpha_{q-2}^k & \alpha_{q-1}^k
\end{array}
\right],
\end{eqnarray*}
where $2<k < q-1$. Let $\cC_k$ be the linear code over $\gf_q$ generated by $M_k$.
Then $\cC_k$ is a $[q-1,k,q-1-k]$ NMDS code and the minimum weight codewords of both $\cC_k$ and its dual $\cC_k^\perp$ support 2-designs.
\end{conj}

By Magma, Conjecture \ref{conj-1}  has been verified to be correct for
$$(m,k_1,k)\in\{(4,4,4),(4,4,5),\ldots,(4,4,12)\}$$
and
$$(m,k_1,k)\in \{(5,4,5),(5,4,6),(5,4,7),(5,4,8)\}.$$
\subsection{The extended code $\overline{\cC_{(1,2,3,4)}}$ of $\cC_{(1,2,3,4)}$}
In this subsection, we study the extended code $\overline{\cC_{(1,2,3,4)}}$ of $\cC_{(1,2,3,4)}$.
It is obvious that the linear code $\overline{\cC_{(1,2,3,4)}}$ is generated by the following matrix:

\begin{eqnarray*}\label{eqn-extended2}
\overline{H_{(1,2,3,4)}}=
\left[
\begin{array}{lllll}
1 & 1 & \cdots & 1  & 1\\
\alpha_1 & \alpha_2 & \cdots & \alpha_{q-1} & 0\\
\alpha_1^2 & \alpha_2^2 & \cdots & \alpha_{q-1}^2 & 0\\
\alpha_1^3 & \alpha_2^3 & \cdots & \alpha_{q-1}^3 & 0\\
\alpha_1^4 & \alpha_2^4 & \cdots & \alpha_{q-1}^4 & 0\\
\alpha_1^6 & \alpha_2^6 & \cdots & \alpha_{q-1}^6 & 0
\end{array}
\right].
\end{eqnarray*}

The following lemma will be used to give our main result in this subsection.
\begin{lemma}\label{thm-prooft8}
Let $m$ be an integer with $m>3$ and $q=2^m$. Let $x_1, x_2, x_3, x_4, x_5, x_6$ be six pairwise distinct elements in $\gf_q$. Define the matrix
\begin{eqnarray*}\label{eqn-m8}
\overline{M_{(1,2,3,4)}}=
\left[
\begin{array}{llllll}
1 & 1 & 1 & 1 & 1 & 1\\
x_1 & x_2 & x_3 & x_4 & x_5 & x_6\\
x_1^2 & x_2^2 & x_3^2 & x_4^2 & x_5^2 & x_6^2\\
x_1^3 & x_2^3 & x_3^3 & x_4^3 & x_5^3 & x_6^3\\
x_1^4 & x_2^4 & x_3^4 & x_4^4 & x_5^4 & x_6^4\\
x_1^6 & x_2^6 & x_3^6 & x_4^6 & x_5^6 & x_6^6
\end{array}\right].
\end{eqnarray*}
Then for any pairwise  different and fixed elements $x_1, x_2, x_3$, the total number of different choices of $(x_4, x_5,x_6)$ such that $|\overline{M_{(1,2,3,4)}}|=0$ is equal to $\frac{(q-4)(q-8)}{6}$ (regardless of the ordering of $x_4, x_5,x_6$).
\end{lemma}

\begin{proof}
Similarly to the proof of Lemma \ref{thm-prooft4}, we can easily derive this lemma by Lemma \ref{lem-Vandermonde}.
\end{proof}

\begin{theorem}\label{thm-8}
Let $m$ be an integer with $m>3$ and $q=2^m$. Then $\overline{\cC_{(1,2,3,4)}}$ is a $[q, 6, q-6]$ NMDS code over $\gf_q$ with weight enumerator
$$A(z)=1 + \frac{q(q-1)^2(q-2)(q-4)(q-8)}{720} z^{q-6} + \frac{q(q-1)^2(q-2)(q-4)}{24} z^{q-5} + $$
$$\frac{q(q-1)^2(q-2)(q^2-2q+2)}{48} z^{q-4}+\frac{q(q-1)^2(q-2)(2q^2+9q+28)}{36} z^{q-3}+$$
$$\frac{q(q-1)^2(9q^3+22q^2+12q+176)}{48} z^{q-2}+\frac{q(q-1)(44q^4+65q^3+125q^2-170q+536)}{120}z^{q-1}$$
$$+\frac{(q-1)(53q^5+27q^4+2q^3+90q^2-172q+144)}{144}z^{q}.  $$
Moreover, the minimum weight codewords of $\overline{\cC_{(1,2,3,4)}}$ support a 3-$(q,q-6,\frac{(q-4)(q-6)(q-7)(q-8)^2}{720})$ simple design and the minimum weight codewords of $\overline{\cC_{(1,2,3,4)}}^\perp$ support a 3-$(q,6,\frac{(q-4)(q-8)}{6})$ simple design.
\end{theorem}

\begin{proof}
By Lemma \ref{thm-prooft8}, we can derive the desired conclusion with a similar proof as that of Theorem \ref{thm-1}. The details are omitted.
\end{proof}

\subsection{When $(i,j,k,l)=(1,2,4,5)$}
In this subsection, we consider the case for $(i,j,k,l)=(1,2,4,5)$.

The following lemma plays an important role in the proof of our next main result.
\begin{lemma}\label{lem-prooft9}
Let $m$ be an odd integer with $m>3$, $q=2^m$. Let $x_1, x_2, x_3, x_4, x_5, x_6$ be six pairwise distinct elements in $\gf_q^*$. Define the matrix
\begin{eqnarray*}\label{eqn-m8}
M_{(1,2,4,5)}=
\left[
\begin{array}{llllll}
1 & 1 & 1 & 1 & 1 & 1\\
x_1 & x_2 & x_3 & x_4 & x_5 & x_6\\
x_1^2 & x_2^2 & x_3^2 & x_4^2 & x_5^2 & x_6^2\\
x_1^4 & x_2^4 & x_3^4 & x_4^4 & x_5^4 & x_6^4\\
x_1^5 & x_2^5 & x_3^5 & x_4^5 & x_5^5 & x_6^5\\
x_1^6 & x_2^6 & x_3^6 & x_4^6 & x_5^6 & x_6^6
\end{array}\right].
\end{eqnarray*}
Then for any two different and fixed elements $x_1, x_2$, the total number of different choices of $(x_3, x_4, x_5, x_6)$ such that $|M_{(1,2,4,5)}|=0$ is equal to $\frac{(q-5)(q-8)}{6}$ (regardless of the ordering of $x_3, x_4, x_5, x_6$).
\end{lemma}

\begin{proof}
Similarly to the proof of Lemma \ref{thm-prooft6},
we can easily derive this lemma by Lemma \ref{lem-Vandermonde}.
\end{proof}

\begin{theorem}\label{thm-9}
Let $m$ be an odd integer with $m>3$ and $q=2^m$. Then $\cC_{(1,2,4,5)}$ is a $[q-1, 6, q-7]$ NMDS code over $\gf_q$ with weight enumerator
$$A(z)=1 + \frac{(q-1)^2(q-2)(q-5)(q-6)(q-8)}{720} z^{q-7} + \frac{(q-1)^2(q-2)(q-5)(7q-36)}{120} z^{q-6} + $$
$$\frac{(q-1)^2(q-2)(q^3-7q^2+34q-96)}{48} z^{q-5}+\frac{(q-1)^2(q-2)(2q^3+7q^2-19q-30)}{36} z^{q-4}+$$
$$\frac{(q-1)^2(q-2)(9q^3+29q^2+62q+240)}{48} z^{q-3}+\frac{(q-1)^2(44q^4+111q^3+219q^2-34q+1320)}{120}z^{q-2} $$
$$+\frac{(q-1)(265q^5+398q^4+417q^3+1108q^2-1708q+3840)}{720}z^{q-1}.  $$
Moreover, the minimum weight codewords of $\cC_{(1,2,4,5)}$ support a 2-$(q-1,q-7,\frac{(q-5)(q-6)(q-7)(q-8)^2}{720})$ simple design and the minimum weight codewords of $\cC_{(1,2,4,5)}^\perp$ support a 2-$(q-1,6,\frac{(q-5)(q-8)}{6})$ simple design.
\end{theorem}

\begin{proof}
Similarly to the proof of Theorem \ref{thm-1}, we can prove this theorem by Lemma \ref{lem-prooft9}.
\end{proof}

\subsection{The extended code $\overline{\cC_{(1,2,4,5)}}$}
It is obvious that the extended  code $\overline{\cC_{(1,2,4,5)}}$ of $\cC_{(1,2,4,5)}$ is generated by the following matrix:

\begin{eqnarray*}\label{eqn-construction-extended3}
\overline{H_{(1,2,4,5)}}=\left[
\begin{array}{llllll}
1 & 1 & \cdots & 1  & 1 & 1\\
\alpha_1 & \alpha_2 & \cdots & \alpha_{q-2} & \alpha_{q-1} &0\\
\alpha_1^2 & \alpha_2^2 & \cdots & \alpha_{q-2}^2 & \alpha_{q-1}^2 & 0\\
\alpha_1^4 & \alpha_2^4 & \cdots & \alpha_{q-2}^4 & \alpha_{q-1}^4 & 0\\
\alpha_1^5 & \alpha_2^5 & \cdots & \alpha_{q-2}^5 & \alpha_{q-1}^5 & 0\\
\alpha_1^6 & \alpha_2^6 & \cdots & \alpha_{q-2}^6 & \alpha_{q-1}^6 & 0
\end{array}
\right].
\end{eqnarray*}

We need the following lemma to give our main result in this subsection.
\begin{lemma}\label{lem-proof-extended3}
Let $m$ be an odd integer with $m>3$ and $q=2^m$. Let $x_1, x_2, x_3, x_4, x_5, x_6$ be six pairwise distinct elements in $\gf_q$. Define the matrix
\begin{eqnarray*}\label{eqn-m9}
\overline{M_{(1,2,4,5)}}=
\left[
\begin{array}{llllll}
1 & 1 & 1 & 1 & 1 & 1\\
x_1 & x_2 & x_3 & x_4 & x_5 & x_6\\
x_1^2 & x_2^2 & x_3^2 & x_4^2 & x_5^2 & x_6^2\\
x_1^4 & x_2^4 & x_3^4 & x_4^4 & x_5^4 & x_6^4\\
x_1^5 & x_2^5 & x_3^5 & x_4^5 & x_5^5 & x_6^5\\
x_1^6 & x_2^6 & x_3^6 & x_4^6 & x_5^6 & x_6^6
\end{array}\right].
\end{eqnarray*}
Then for any three pairwise different and fixed elements $x_1, x_2, x_3$, the total number of different choices of $(x_4, x_5, x_6)$ such that $|\overline{M_{(1,2,4,5)}}|=0$ is equal to $\frac{(q-5)(q-8)}{6}$ (regardless of the ordering of $x_4, x_5, x_6$).
\end{lemma}

\begin{proof}
Similarly to the proof of Lemma \ref{thm-prooft6},
we can easily derive this lemma by Lemma \ref{lem-Vandermonde}.
\end{proof}

\begin{theorem}\label{thm-10}
Let $m$ be an odd integer with $m>3$ and $q=2^m$. Then $\overline{\cC_{(1,2,4,5)}}$ is a $[q, 6, q-6]$ NMDS code over $\gf_q$ with weight enumerator
$$A(z)=1 + \frac{q(q-1)^2(q-2)(q-5)(q-8)}{720} z^{q-6} + \frac{q(q-1)^2(q-2)(3q-14)}{60} z^{q-5} + $$
$$\frac{q(q-1)^2(q-2)(q^2-3q+10)}{48} z^{q-4}+\frac{q(q-1)^2(q-2)(q^2+5q+10)}{18} z^{q-3}+$$
$$\frac{q(q-1)^2(9q^3+21q^2+22q+160)}{48} z^{q-2}+\frac{q(q-1)(22q^4+33q^3+57q^2-72q+260)}{60}z^{q-1} $$
$$+\frac{(q-1)(265q^5+134q^4+21q^3+424q^2-844q+720)}{720}z^{q}.  $$
Moreover, the minimum weight codewords of $\overline{\cC_{(1,2,4,5)}}$ support a 3-$(q,q-6,\frac{(q-5)(q-6)(q-7)(q-8)^2}{720})$ simple design and the minimum weight codewords of $\overline{\cC_{(1,2,4,5)}}^\perp$ support a 3-$(q,6,\frac{(q-5)(q-8)}{6})$ simple design.
\end{theorem}

\begin{proof}
Similarly to the proof of Theorem \ref{thm-1}, we can prove this theorem by Lemma \ref{lem-proof-extended3}.
\end{proof}

\section{Optimal locally recoverable codes}\label{sect-LRC}
In this section, we study the minimum locality of the codes constructed in this paper.
Locally recoverable codes (LRCs for short) are applied in distributed storage and cloud storage. LRCs were proposed for the recovery of data by  Gopalan, Huang, Simitci and Yikhanin \cite{LRC}. Let $[n]:=\{1,2,\cdots,n\}$ for a positive integer $n$. Let $\cC$ be an $[n,k,d]$ linear code over $\gf_q$. For every $i \in [n]$, if there exist a subset $R_i\subseteq [n]\backslash \{i\}$ of size $r$ $(r<k)$ and a function $f_i(x_1,x_2,\cdots,x_r)$ on $\gf_q^r$ such that $c_i=f_i(\bc_{R_i})$ for each $\bc=(c_1,c_2,\cdots,c_{n})$ in $\cC$, then $\cC$ is called an $(n,k,d,q;r)$-LRC, where $\bc_{R_i}$ is the projection of $\bc$ at $R_i$.  $R_i$ is said to be the recovering set of $c_i$.  The minimum $r$ such that a linear code $\cC$ is an $(n,k,d,q;r)$-LRC is called the minimum locality of this code.

\begin{lemma}\cite[\textbf{Singleton-like bound}]{LRC}\label{lem-d}
For any $(n,k,d,q;r)$-LRC, we have
\begin{eqnarray*}\label{eqn-Slbound}
d \leq n-k-\left\lceil\frac{k}{r}\right\rceil+2.
\end{eqnarray*}
\end{lemma}

A LRC is said to be distance-optimal ($d$-optimal for short) if it achieves the Singleton-like bound.

\begin{lemma}\cite[\textbf{Cadambe-Mazumdar bound}]{CM}\label{lem-k}
For any $(n,k,d,q;r)$-LRC,
\begin{eqnarray*}\label{eqn-CMbound}
k \leq \min_{t\in \Bbb Z^+}\left[rt+k_{opt}^{(q)}\left(n-t(r+1),d\right)\right],
\end{eqnarray*}
where $\Bbb Z^{+}$ denotes the set of all positive integers, $k_{opt}^{(q)}(n,d)$ is the the largest possible dimension of a linear code with alphabet size $q$, length $n$, and minimum  distance $d$.
\end{lemma}

A LRC is said to be dimension-optimal ($k$-optimal for short) if it achieves the Cadambe-Mazumdar bound.

 Let $\mathcal{B}_i(\mathcal{C})$ denote the set of supports of all codewords with Hamming weight $i$ in $\mathcal{C}$. The following lemma is useful for determining the minimum locality for nontrivial linear codes whose minimum distance is larger than 1.

\begin{lemma}\label{lem-tlrc}\cite{TANP}
Let $\cC$ be a nontrivial linear code of length $n$ and put $d^\perp =d(\cC^\perp).$ If $(\mathcal{P},\mathcal{B}_{d^\perp}(\mathcal{C^\perp}))$ is a 1-$(n, d^\perp, \lambda_{1}^\perp )$ design with $\lambda_{1}^\perp \geq 1$, then $\cC$ has minimum locality $d^\perp-1$.
\end{lemma}

\begin{theorem}\label{th-LLRC}
All the locally repairable codes listed in table \ref{tab-lrc} are both $d$-optimal and $k$-optimal.
\end{theorem}

\begin{table}[ht]
\begin{center}
\caption{Optimal locally repairable codes}\label{tab-lrc}
\begin{tabular}{cc} \hline
Optimal locally repairable codes & parameters \\ \hline
$\cC_D$,$\cC_H$  & $(q-1, 3, q-4, q; 2)$  \\
$\cC_D^\perp$,$\cC_H^\perp$  & $(q-1, q-4,3,q;q-5)$  \\
$\cC_{(1,3)}$,$\cC_{(2,3)}$  & $(q-1, 4, q-5, q; 3 )$  \\
$\cC_{(1,3)}^\perp$,$\cC_{(2,3)}^\perp$  & $(q-1, q-5, 4, q; q-6 )$  \\
$\cC_{(2,3,4)}$,$\cC_{(1,2,3)}$,$\cC_{(1,2,4)}$,$\cC_{(1,3,4)}$  & $(q-1, 5, q-6, q; 4)$  \\
$\cC_{(2,3,4)}^\perp$,$\cC_{(1,2,3)}^\perp$,$\cC_{(1,2,4)}^\perp$,$\cC_{(1,3,4)}^\perp$  & $(q-1, q-6, 5, q; q-7)$  \\
$\overline{\cC_{(1,2,4)}}$  &  $(q, 5, q-5, q; 4)$  \\
$\overline{\cC_{(1,2,4)}}^\perp$  &  $(q, q-5, 5, q; q-6)$  \\
$\cC_{(1,2,3,4)}$, $\cC_{(2,3,4,5)}$, $\cC_{(1,2,4,5)}$   &  $(q-1, 6, q-7, q; 5)$  \\
$\cC_{(1,2,3,4)}^\perp$, $\cC_{(2,3,4,5)}^\perp$, $\cC_{(1,2,4,5)}^\perp$   &  $(q-1, q-7, 6, q; q-8)$  \\
$\overline{\cC_{(1,2,3,4)}}$, $\overline{\cC_{(1,2,4,5)}}$   &  $(q, 6, q-6, q; 5 )$  \\
$\overline{\cC_{(1,2,3,4)}}^\perp$, $\overline{\cC_{(1,2,4,5)}}^\perp$   &  $(q, q-6, 6, q; q-7 )$  \\
\hline
\end{tabular}
\end{center}
\end{table}

\begin{proof}
Note that all NMDS codes in Sections \ref{sect-code1}-\ref{sect-code4} hold 2-designs or 3-designs.
Then the desired conclusion follows from Lemmas \ref{lem-d}, \ref{lem-k} and \ref{lem-tlrc}.
\end{proof}

In \cite{LRC02, NMDS1, NMDS2, TANP}, optimal locally repairable codes of distances 3 and 4 were constructed. We remark that the optimal locally repairable codes of distance 3 and 4 in this paper have different parameters from those in these papers.
\section{Summary and concluding remarks}

In this paper, we presented several infinite families of near MDS codes holding $t$-designs over $\gf_{2^m}$. Our main contributions are listed in the following:
\begin{enumerate}
\item In Section \ref{sect-code1}, two families of $[q-1,3,q-4]$ NMDS codes holding 2-designs were constructed;
\item In Section \ref{sect-code2}, two families of $[q-1,4,q-5]$ NMDS codes holding 2-designs were constructed;
\item In Section \ref{sect-code3},  four families $[q-1,5,q-6]$ NMDS codes holding 2-designs and a family of $[q,5,q-5]$ NMDS codes holding 3-designs were constructed;
\item In Section \ref{sect-code4},  three families $[q-1,6,q-7]$ NMDS codes holding 2-designs and two families of $[q,6,q-6]$ NMDS codes holding 3-designs were constructed.
\end{enumerate}
We remark that only the first two families of NMDS codes satisfy the Assmus-Mattson Theorem.
Though other families of NMDS codes do not satisfy the Assmus-Mattson Theorem, they still hold $t$-designs.
Besides,  all NMDS codes constructed in this paper are both $d$-optimal and $k$-optimal LRCs.

Constructing NMDS codes holding $t$-designs for $t\geq 2$ is very challenging. All known infinite  families of NMDS codes holding $t$-designs have small dimensions.
It is open to construct NMDS codes holding $t$-designs with general dimensions.
If Conjecture \ref{conj-1} in this paper is true, then this question can be tackled.
The reader is invited to prove our conjecture.




\end{document}